\newif\iflncs
\lncsfalse

\iflncs
\documentclass{llncs}
\else
\documentclass[11pt]{article}
\fi

\usepackage{multirow}
\usepackage{makecell}
\usepackage{booktabs}
\usepackage{afterpage}

\usepackage{epigraph}
\usepackage{graphicx}
\usepackage{tikz}
\usetikzlibrary{arrows,automata,patterns,decorations.pathreplacing}
\usepackage{enumitem}
\usepackage[hidelinks,hypertexnames=false]{hyperref}
\usepackage{lipsum}

\usepackage[linesnumbered,ruled,vlined,boxed,algo2e]{algorithm2e}
\usepackage{nicefrac}
\usepackage{amsmath}
\usepackage{float}
\usepackage{mathtools}
\usepackage{amsfonts}
\usepackage{clipboard}
\iflncs\else
\usepackage[margin=1in]{geometry}
\usepackage{amsthm}
\fi
\usepackage{thm-restate}
\usepackage[nameinlink]{cleveref}
\iflncs

\spnewtheorem{claim}{Claim}[lemma]{\bfseries}{\itshape}
\spnewtheorem{subclaim}{Claim}[claim]{\bfseries}{\itshape}

\else

\newtheorem{lemma}{Lemma}

\newtheorem{property}{Property}
\newtheorem{claim}{Claim}[lemma]

\fi

\crefname{property}{property}{properties}
\Crefname{property}{Property}{Properties}
\crefname{claim}{claim}{claims}
\Crefname{claim}{Claim}{Claims}
\crefname{subclaim}{claim}{claims}
\Crefname{subclaim}{Claim}{Claims}
\Crefname{enumi}{(item)}{(items)}

\newcommand{\FullRef}[1]{\iflncs the full version~\cite{KadriaR26sssp}\else\Cref{#1}\fi}

\newenvironment{reminder}[1]{\smallskip
\noindent {\bf #1. }\em}{}
\usepackage{xcolor}
\definecolor{DarkGreen}{RGB}{1,50,32}
\usepackage{pgfplots}
\pgfplotsset{compat=1.18}

\usepackage{silence}
\usepackage{chngcntr}
\counterwithin*{equation}{section}
\counterwithin*{equation}{subsection}

\begin{document}
\ActivateWarningFilters[pdftoc]
\newcommand{\defeq}{:=}
\newcommand{\eps}{\varepsilon}

\newcommand{\liam}[1]{{\color{red} \textbf{Liam}: #1}} 
\newcommand{\avi}[1]{{\color{purple} \textbf{Avi}: #1}} 
\newcommand{\new}[1]{{\color{blue} #1}}

\newcommand{\blue}[1]{{\color{blue}#1}}

\newcommand{\ReturnCode}{\textbf{return}}
\newcommand{\codestyle}[1]{\texttt{#1}\xspace}
\newcommand{\Initialize}{\mbox{\codestyle{Initialize}}}
\newcommand{\ImprovedSSSPBalanced}{\codestyle{ImprovedSSSPBalanced}}
\newcommand{\ImprovedSSSP}{\codestyle{ImprovedSSSP}}
\newcommand{\Cycle}{\mbox{\codestyle{Cycle}}}
\newcommand{\ADO}{\mbox{\codestyle{ADO}}}
\newcommand{\hADO}{\mbox{\codestyle{hADO}}}
\newcommand{\Query}{\mbox{\codestyle{.Query}}}
\newcommand{\ADOQuery}{\mbox{\codestyle{ADO.Query}}}
\newcommand{\ConstructADO}{\codestyle{ConstructADO}}
\newcommand{\FastAPSP}{\codestyle{FastAPSP}}
\newcommand{\Construct}{\mbox{\codestyle{Construct}}}
\newcommand{\Spanner}{\codestyle{Spanner}}
\newcommand{\RT}{\mbox{\codestyle{RT}}}
\renewcommand{\L}{\mbox{\codestyle{L}}}
\newcommand{\CycleOdd}{\codestyle{CycleOdd}}
\newcommand{\BallOrCycle}{\codestyle{BallOrCycle}}
\newcommand{\ClusterOrCycleBounded}{\codestyle{ClusterOrCycleBounded}}
\newcommand{\ClusterOrCycle}{\codestyle{ClusterOrCycle}}
\newcommand{\SimpleCycle}{\codestyle{SimpleCycle}}
\newcommand{\Next}{\codestyle{Next}}
\newcommand{\Sample}{\codestyle{Sample}}
\newcommand{\Dijkstra}{\codestyle{Dijkstra}}
\newcommand{\BundleDijkstra}{\codestyle{BundleDijkstra}}
\newcommand{\ActiveBellmanFord}{\codestyle{ActiveBellmanFord}}
\newcommand{\BoundedBellmanFord}{\codestyle{BoundedBellmanFord}}
\newcommand{\BoundSSSP}{\codestyle{BoundSSSP}}
\newcommand{\Preprocess}{\codestyle{Preprocess}}
\newcommand{\HashTable}{\codestyle{HashTable}}
\newcommand{\Heap}{\codestyle{Heap}}
\newcommand{\RelaxNext}{\codestyle{RelaxNext}}

\newcommand{\PreprocessGraph}{\codestyle{Initialize}}
\newcommand{\Route}{\codestyle{Route}}
\newcommand{\TreeRoute}{\codestyle{TreeRoute}}
\newcommand{\N}{\mathbb{N}}
\newcommand{\MinCycle}{\codestyle{MinCycle}}
\newcommand{\Ball}{\textnormal{\texttt{Ball}}\xspace}
\newcommand{\Bundle}{\textnormal{\texttt{Bundle}}\xspace}
\newcommand{\DistanceOracle}{\codestyle{TZ-DistanceOracle}}
\newcommand{\SparseOrCycle}{\codestyle{SparseOrCycle}}
\newcommand{\Intersection}{\codestyle{Intersection}}
\newcommand{\CycleAdditive}{\codestyle{CycleAdditive}}
\newcommand{\GenerateSi}{\codestyle{ComputeS}}

\newcommand{\codeNull}{\codestyle{null}}
\newcommand{\codeYes}{\codestyle{Yes}}
\newcommand{\codeNo}{\codestyle{No}}
\newcommand{\codeAnd}{~ \mathrm{and} ~}
\newcommand{\codeOr}{~ \mathrm{or} ~}
\newcommand{\wt}{\ell}
\newcommand{\id}{\mathrm{id}}
\newcommand{\BMSSP}{\codestyle{BMSSP}}
\newcommand{\Cl}{CL}
\newcommand{\CL}{CL}
\newcommand{\cl}{c\ell}

\newcommand{\EQ}{\;=\;}
\newcommand{\GE}{\;\ge\;}
\newcommand{\Ot}{\tilde{O}}
\newcommand{\stactri}{\stackrel\triangle}

\newcommand{\EE}{\mathbb{E}}
\newcommand{\RR}{\mathbb{R}}

\newcommand{\hop}[1]{^{\mspace{2mu}\raisebox{0.75ex}{\ensuremath{\scriptstyle\le #1}}}}

\DeclarePairedDelimiter{\ceil}{\lceil}{\rceil}
\DeclarePairedDelimiter{\floor}{\lfloor}{\rfloor}
\DeclarePairedDelimiter{\pair}{\langle}{\rangle}

\title{A Faster Undirected Single-Source Shortest Path Algorithm}
\iflncs
\else
\date{}
\author{Avi Kadria\thanks{Department of Computer Science, Bar-Ilan University, Ramat Gan 5290002, Israel. E-mail \texttt{avi.kadria3@gmail.com}.} \and Liam Roditty\thanks{Department of Computer Science, Bar-Ilan University, Ramat Gan 5290002, Israel. E-mail \texttt{liam.roditty@biu.ac.il}. Supported in part by BSF grants 2016365 and 2020356.}}
\fi
\maketitle

\begin{abstract}

The single-source shortest paths (SSSP) problem in graphs with non-negative edge weights is one of the most classic problems in algorithms. For decades, the best known running time in the comparison-addition model was the $O(m+n\log n)$ bound of Dijkstra's algorithm with Fibonacci heaps.

Recently, Duan, Mao, Shu, and Yin (FOCS'23) gave a randomized $O(m\log^{1/2} n (\log\log{n})^{1/2})$-time algorithm for SSSP in weighted undirected graphs. For weighted directed graphs, Duan, Mao, Mao, Shu, and Yin (STOC'25) gave an $O(m\log^{2/3} n)$-time algorithm for SSSP. Very recently, Duan, Mao, Shu, and Yin (ICALP'26) obtained an algorithm for directed graphs whose running time matches the $O(m\log^{1/2} n (\log\log{n})^{1/2})$ time of the undirected case.

In this paper, we present a faster algorithm for SSSP in weighted undirected graphs, giving the first improvement in running time since the FOCS'23 breakthrough of Duan, Mao, Shu, and Yin. Our algorithm runs in $O(m\log^{1/2} n (\log\log{n})^{1/4} (\log\log\log{n})^{1/4})$ time, improving the previous running time by a factor of $(\frac{\log\log n}{\log\log\log n})^{1/4}$. 

Our main contribution is a simple and efficient tool that computes, for every vertex, its distance to the nearest vertex in a random sample; this tool may be of independent interest.

\end{abstract}

\section{Introduction}
Computing single-source shortest paths (SSSP) in a weighted graph $G=(V,E,\wt)$ is one of the most classic problems in algorithms, studied for over half a century~\cite{bellman1958routing,Di59,moore1959shortest,Spira73}.
Until recently, the best known running time, in the comparison-addition model, was $O(m+n\log{n})$ using Dijkstra's famous algorithm with Fibonacci heaps~\cite{fredman1987fibonacci}.\footnote{For bounded integer weights, faster algorithms have been known for more than two decades~\cite{dial1969algorithm,ahuja1990faster,thorup1999undirected,hagerup2000improved,thorup2004integer,pettie2005shortest}.}

The $n\log n$ term in this bound seems hard to avoid: Dijkstra's algorithm visits the vertices in increasing order of distance, and hence produces a sorted ordering of the distances, a task that requires $\Omega(n\log n)$ comparisons in the worst case. Moreover, Haeupler, Hlad\'ik, Rozho\v{n}, Tarjan, and T\v{e}tek~\cite{haeupler2024universal} showed that Dijkstra's algorithm (with a suitable heap) is \emph{universally optimal} for this ordering task. 
Thus, breaking the so-called \emph{sorting barrier} requires a fundamentally different approach: computing the distances without sorting them.

Since $\Omega(m)$ time is necessary even for graph reachability, the $m$ term in Dijkstra's running time is unavoidable.
All recent work on this question has focused on the sparse-graph regime, where $m=O(n)$ and the $n\log n$ term dominates the running time. Therefore, throughout the paper, we assume $m=O(n)$ and discuss general $m$ in \Cref{S-discussion}.

The first to break the sorting barrier were Duan, Mao, Shu, and Yin~\cite{duan2023randomized}, who gave a randomized $O(m\log^{1/2}{n}(\log\log{n})^{1/2})$ time algorithm for sparse weighted \emph{undirected} graphs.\footnote{For general $m$ the bound is refined: $O(n\sqrt{\log n\log\log n})$ when $m\le n\log\log n$, and $O(\sqrt{mn\log n})$ when $n\log\log n\le m\le n\log n$.} 
Recently, Yan~\cite{yan2026lossless} derandomized this algorithm and achieved the same running time deterministically.

In a subsequent breakthrough, Duan, Mao, Mao, Shu, and Yin~\cite{duan2025breaking}\footnote{Winner of the STOC'25 Best Paper Award.} broke the sorting barrier for sparse directed graphs as well and presented a complex recursive algorithm for SSSP that runs in $O(m\log^{2/3}{n})$ time.
Very recently, this result was improved by Duan, Mao, Shu, and Yin~\cite{duan2026faster} to $O(m\log^{1/2}{n}(\log\log{n})^{1/2})$ time,\footnote{For general $m$ the bound is $O(m\sqrt{\log n}+\sqrt{mn\log n\log\log n})$.} which matches the running time in undirected graphs \cite{duan2023randomized}.

\paragraph{Our Results.}
In this paper, we improve the state-of-the-art running time for SSSP in sparse weighted undirected graphs, giving the first improvement in running time since~\cite{duan2023randomized}. 
\begin{restatable}{theorem}{theoremImprovedUndirected}\label{T-Improved-undirected}
    There is a randomized algorithm that w.h.p.\ runs in
    \[O(m\log^{1/2}{n}(\log\log{n})^{1/4}(\log\log\log{n})^{1/4})\]
    time and computes single-source shortest paths in undirected graphs with non-negative edge weights.
\end{restatable}
The algorithm of \Cref{T-Improved-undirected} is faster than the algorithm of~\cite{duan2023randomized} by a factor of $(\frac{\log\log{n}}{\log\log\log{n}})^{1/4}$. 
We summarize the previous results alongside our new result in \Cref{tab:comparison}.

\afterpage{%
\begin{table}[t]
\caption{Comparison of algorithms for SSSP in sparse ($m=O(n)$) weighted graphs in the comparison-addition model.}
\label{tab:comparison}
\centering
\begin{tabular}{llll}
\toprule
Reference & Graph type & Randomized? & Running time \\
\midrule
Dijkstra & directed & deterministic & $O(m+n\log{n})$ \\
\cite{duan2023randomized} & undirected & randomized & $O(m\log^{1/2}{n}(\log\log{n})^{1/2})$ \\
\cite{duan2025breaking} & directed & deterministic & $O(m\log^{2/3}{n})$ \\
\cite{yan2026lossless} & undirected & deterministic & $O(m\log^{1/2}{n}(\log\log{n})^{1/2})$ \\
\cite{duan2026faster} & directed & deterministic & $O(m\log^{1/2}{n}(\log\log{n})^{1/2})$ \\
\textbf{This paper} & undirected & randomized & $O(m\log^{1/2}{n}(\log\log{n})^{1/4}(\log\log\log{n})^{1/4})$ \\
\bottomrule
\end{tabular}
\end{table}%
}

\paragraph{Techniques.}
The algorithm of Duan, Mao, Shu, and Yin~\cite{duan2023randomized} is built on two data structures with respect to a random sample $S\subseteq V$ of size $\Theta(n/k)$, where $k$ is a parameter of the algorithm, eventually set to roughly $\log^{1/2}{n}$.
Let $b(v)$ be the nearest sample to $v$.
The first data structure is the \emph{bundle} of a sample $s\in S$, defined as $\Bundle(s)=\{v\in V:b(v)=s\}$, the set of vertices whose nearest sample is $s$.
The second data structure is the \emph{ball} of a vertex $v$, consisting of its nearest sample $b(v)$ and all vertices closer to $v$ than $b(v)$ (see \Cref{fig:bundle-ball} for an illustration).
In~\cite{duan2023randomized}, both data structures are computed by running a bounded Dijkstra from each vertex until the first sample is reached; this costs $O(mk\log k)$ time over all vertices.

Surprisingly, the dominant cost of~\cite{duan2023randomized} is this preprocessing step, and not the shortest-path computation that follows.
The extra $\log k$ factor is precisely what produces the $(\log\log{n})^{1/2}$ term in the running time, and improving this factor is the key to our result.


Instead of $n$ local bounded Dijkstra searches, each terminating when the nearest sample is reached, we run a single global active-vertex Bellman--Ford algorithm~\cite{moore1959shortest} from an auxiliary super-source connected to all samples, which simultaneously computes all distances to $S$.
One might ask why such a simple change was not already made in~\cite{duan2023randomized}: the reason it appears not to work is that a random sample leaves some balls of size $\Theta(k\log n)$ w.h.p. Therefore, running Bellman--Ford from a super-source takes $O(mk\log n)$ time, a $\log n$ factor above our goal.
Our main technical contribution is to show that if one restricts the relaxation only to \emph{active} vertices then w.h.p.\ the total running time of the algorithm is $O(mk)$, improving upon both Bellman--Ford ($O(mk\log n)$) and Dijkstra ($O(m+n\log n)$) from a super-source.

\begin{restatable}[Sorting-free distances to a sampled set]{theorem}{thmActiveBF}\label{T-active-BF}
    Let $G$ be a weighted undirected graph with maximum degree $O(m/n)$ (guaranteed by a standard degree reduction, \Cref{S-degree-reduction}) and let $S\subseteq V$ be any non-empty set. Active-vertex Bellman--Ford\iflncs\else{} (\Cref{alg:active-bf})\fi{} computes $\delta(v,b(v))$ and $\Bundle(s)$ for every $v\in V$ and $s\in S$ in $O\!\left(\sum_{v\in V}|\Ball(v)|\cdot \tfrac{m}{n}\right)$ time.
\end{restatable}
For a random sample, the sum of all ball sizes is $O(nk)$ with high probability (see \FullRef{L-sampling}), giving the $O(mk)$ running time.

Using the computed distances to $S$ as radii, the bounded-distance Dijkstra of~\cite{duan2026faster} then computes all balls $\{\Ball(v)\}_{v\in V}$ in total time
\[O(mk\log^{1/2}k(\log\log{k})^{1/2}),\] yielding \Cref{T-Improved-undirected}.
With the bundle construction now costing only $O(mk)$, the ball construction becomes the bottleneck of the entire algorithm, and any further improvement to it immediately improves the overall running time.
In particular, constructing all balls in $O(mk)$ total time would yield an $O(m\log^{1/2}{n})$ time algorithm for SSSP (see \Cref{S-discussion}).

\paragraph{A Sorting-Free Tool of Independent Interest.}
Active-vertex Bellman--Ford, the queue-based variant of Bellman--Ford~\cite{moore1959shortest} that is widely known as SPFA (Shortest Path Faster Algorithm)~\cite{duan1994spfa}, is \emph{not} a fast algorithm in general: on arbitrary inputs it can take $\Theta(mn)$ time, far more than Dijkstra's algorithm.
Our key insight is that for distances to a sampled set, each vertex $v$ is activated at most $|\Ball(v)|$ times, so that a random sample keeps the total work at $O(mk)$; to the best of our knowledge, this is the first instance-dependent bound under which this algorithm provably outperforms a Dijkstra-based approach.
Moreover, computing the distances to a sampled set is a basic primitive that lies at the core of many algorithms, for example, the constructions of approximate distance oracles~\cite{DBLP:journals/jacm/ThorupZ05}, compact routing schemes~\cite{DBLP:conf/spaa/ThorupZ01}, and spanners~\cite{DBLP:conf/icalp/RodittyTZ05}; in particular, the balls essentially coincide with (one level of) the \emph{bunches} of Thorup and Zwick~\cite{DBLP:journals/jacm/ThorupZ05}.
Since our bound improves the $O(mk\log k)$ time of the per-vertex bounded Dijkstra used by~\cite{duan2023randomized} to $O(mk)$, we believe this tool may be of independent interest.

\paragraph{Organization.}
\iflncs
\else
\Cref{S-prelim} contains preliminaries. \Cref{S-small-improve} presents the algorithm in two stages: first under a simplifying \emph{balanced case} assumption on the random sample (\Cref{S-construction}), and then in full generality, using active-vertex Bellman--Ford (\Cref{S-Active-BF}). The concentration bounds for a random sample (\Cref{L-sampling}) are proved in \Cref{A-sampling}. We conclude in \Cref{S-discussion} with a discussion of general $m$ and with open problems.
\fi

\section{Preliminaries}\label{S-prelim}
Let $G=(V,E,\wt)$ be a weighted undirected graph with $\wt:E\to\mathbb{R}_{\ge 0}$, $n$ vertices, and $m$ edges.
Throughout the paper, we focus on the case that $m=O(n)$, in which the sorting barrier dominates.
We work in the comparison-addition model: the only allowed operations on weights are additions and comparisons, each taking $O(1)$ time.
We assume that $G$ is connected (hence $m\ge n-1$); this is without loss of generality, since single-source shortest paths are computed only within the connected component of the source, which can be computed in $O(m)$ time.
To avoid degenerate expressions, throughout the paper $\log{x}$ denotes $\max\{\log_2{x},\,1\}$; in particular, $x\mapsto\log^{1/2}(x)(\log\log(x))^{1/2}$ is well defined, at least $1$, and non-decreasing for all $x\ge 1$.
We say that an event holds \emph{with high probability} (w.h.p.) if, for every constant $c\ge 1$, it holds with probability at least $1-O(n^{-c})$.

\subsection{Reduction to Bounded-Degree Graphs}\label{S-degree-reduction}
To obtain a degree bound of $\Delta\in [3,m/n]$, we use the following standard reduction.
Given a graph $G$ with $m$ edges and $n$ vertices, we preprocess $G$ in $O(m)$ time into a graph with $O(m)$ edges, $O(m/\Delta)$ vertices, and maximum degree bounded by $\Delta$, while preserving all shortest path lengths. A classical way to achieve this~\cite{frederickson1983data} works as follows:

For every vertex $v$ with degree $\deg(v) > \Delta$, we replace $v$ with a cycle $C_v$ of $\ceil{\frac{\deg(v)}{\Delta-2}}$ vertices.
Each vertex in $C_v$ has two weight-$0$ edges, to the next and previous vertices in $C_v$, and at most $\Delta-2$ edges to distinct neighbors of $v$, so its degree is at most $\Delta$. This preserves all shortest path lengths. See \Cref{fig:degree-reduction} for an illustration.
This degree bound is crucial for the efficiency of our active-vertex Bellman--Ford algorithm (see \FullRef{S-Active-BF}), as it guarantees that each active vertex triggers only $O(m/n)$ edge relaxations.
Throughout the paper, we assume the input graph has been preprocessed by this reduction with $\Delta=m/n$, so that it has $O(n)$ vertices, $O(m)$ edges, and maximum degree $m/n$. For $m=O(n)$, the maximum degree is constant. We keep the factor $m/n$ in our bounds, as they also hold for general $m$ (see \Cref{S-discussion}).
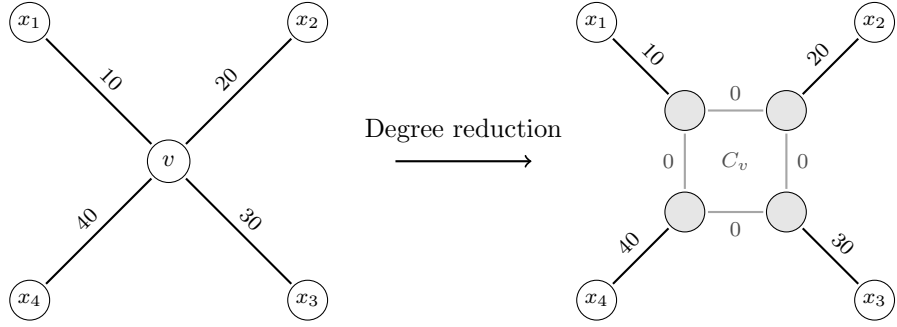
\begin{figure}[t]
\centering
\begin{tikzpicture}[
    vert/.style  = {circle, draw=black, fill=white, minimum size=16pt, inner sep=0pt, font=\footnotesize},
    cvert/.style = {circle, draw=black, fill=gray!20, minimum size=15pt, inner sep=0pt, font=\scriptsize},
    nbr/.style   = {circle, draw=black, fill=white, minimum size=15pt, inner sep=0pt, font=\scriptsize},
    redge/.style = {thick, shorten >=1pt, shorten <=1pt},
    zedge/.style = {thick, gray!70, shorten >=1pt, shorten <=1pt},
]

\begin{scope}
\node[vert] (v) at (0,0) {$v$};
\foreach \i/\ang/\w in {1/135/10, 2/45/20, 3/-45/30, 4/-135/40} {
    \node[nbr] (n\i) at (\ang:2.6) {$x_{\i}$};
    \draw[redge] (v) -- (n\i)
        node[midway, sloped, above, font=\scriptsize] {$\w$};
}
\end{scope}

\begin{scope}[xshift=7.5cm]
\foreach \i/\ang in {1/135, 2/45, 3/-45, 4/-135}
    \node[cvert] (c\i) at (\ang:0.95) {};
\foreach \a/\b/\pos in {1/2/above, 2/3/right, 3/4/below, 4/1/left}
    \draw[-,zedge] (c\a) -- (c\b)
        node[midway, \pos, font=\scriptsize, text=gray!60!black] {$0$};
\node[font=\scriptsize, gray!60!black] at (0,0) {$C_v$};

\node[nbr] (m1) at (135:2.6) {$x_{1}$};
\node[nbr] (m2) at ( 45:2.6) {$x_{2}$};
\node[nbr] (m3) at (-45:2.6) {$x_{3}$};
\node[nbr] (m4) at (-135:2.6) {$x_{4}$};
\foreach \i/\w in {1/10, 2/20, 3/30, 4/40} \draw[-,redge] (c\i) -- (m\i)
    node[midway, sloped, above, font=\scriptsize] {$\w$};

\end{scope}

\draw[->, thick] (3.0,0) -- (4.8,0);
\node[font=\small, above] at (3.9,0.12) {Degree reduction};

\end{tikzpicture}
\caption{%
    Illustration of the degree reduction of \Cref{S-degree-reduction} with $\Delta=3$.
}
\label{fig:degree-reduction}
\end{figure}

\subsection{Total Order on Paths}\label{S-positive-weights}
Distinct paths may have equal weight, especially with the weight-$0$ edges introduced by \Cref{S-degree-reduction}.
We break such ties by the standard infinitesimal perturbation of the weights, without modifying them.
Number the edges of $G$ as $e_1,\dots,e_m$, and conceptually add $2^{i}\eps$ to $\wt(e_i)$, for an infinitesimal $\eps>0$.
That is, every walk $P$ gets the \emph{key}
$$\kappa(P)=\Bigl(\wt(P),\ \sum_{e_i\in P}2^{i}\Bigr),$$
where repeated edges are counted with multiplicity, and we write $P\prec Q$ when $\kappa(P)$ is lexicographically smaller than $\kappa(Q)$.

The key is additive, and every edge, even of weight $0$, has a positive second coordinate.
Hence $\prec$ refines weight, is preserved by appending a common suffix, and makes proper subpaths strictly $\prec$-smaller.
Moreover, a path and its reverse have the same key, while any distinct simple paths have distinct edge sets, and hence distinct keys.

We let $P_G(u,v)$ be the $\prec$-minimum path from $u$ to $v$, and $\delta_G(u,v)$ its weight.
When the graph is clear from context we drop the subscript and write $\delta$ for $\delta_G$ and $P$ for $P_G$.
We write $\delta(u,x)\prec\delta(u,y)$ as shorthand for $P(u,x)\prec P(u,y)$.
All distance comparisons from a common vertex are resolved by $\prec$, ensuring unambiguous choices for $b(v)$ and $\Ball(v)$ in \Cref{S-balls-bundles}, the relaxations in \iflncs\else\Cref{alg:bounded-bf,alg:active-bf}\fi, and the bound $B$ in \BoundSSSP (\Cref{L-Faster-D-SSSP}).

\begin{property}\label{P-shortest-paths}
    Let $u\in V$. Then:
    \begin{enumerate}
        \item\label{P-sp-simple} $P(u,v)$ is uniquely defined and simple (no shortest path traverses a cycle);
        \item\label{P-sp-tree} every prefix of $P(u,v)$ is the $\prec$-minimum path to its endpoint, so $\{P(u,v)\}_{v\in V}$ form a tree rooted at $u$;
        \item\label{P-sp-distinct} distances $\{\delta(u,v)\}_{v\in V}$ are pairwise distinct under $\prec$; in particular $b(v)$ in \Cref{S-balls-bundles} is unique;
        \item\label{P-sp-prefix} if $x\neq v$ lies on $P(u,v)$, then $\delta(u,x)\prec\delta(u,v)$;
        \item\label{P-sp-reverse} $P(v,u)$ is the reverse of $P(u,v)$.
    \end{enumerate}
\end{property}
\begin{proof}
    Removing a cycle from a walk does not increase its weight and strictly decreases its second coordinate, yielding a strictly $\prec$-smaller walk. As distinct simple paths from $u$ have distinct keys, the $\prec$-minimum walk from $u$ to $v$ exists, is unique, and is a simple path, which proves item~\ref{P-sp-simple}.
    For item~\ref{P-sp-tree}, let $P$ be the prefix of $P(u,v)$ to $x$ and $R$ the suffix. If a path $P'$ from $u$ to $x$ satisfies $P'\prec P$, then $P'\circ R\prec P(u,v)$, and removing cycles from $P'\circ R$ only further decreases it under $\prec$, contradicting the minimality of $P(u,v)$.
    Item~\ref{P-sp-distinct} holds because paths from $u$ to distinct endpoints are distinct, hence have distinct keys.
    Item~\ref{P-sp-prefix} follows from item~\ref{P-sp-tree}, as $P(u,x)$ is a proper prefix of $P(u,v)$.
    Finally, item~\ref{P-sp-reverse} holds because reversal maps the paths from $u$ to $v$ bijectively onto the paths from $v$ to $u$ and preserves keys.
\end{proof}

\subsection{Bellman--Ford Algorithm}

We denote by $P_G^{\le k}(u,v)$ the shortest among the paths between $u$ and $v$ that contain at most $k$ edges, and by $\delta_G^{\le k}(u,v)$ its length.
For a path $P$, we denote by $|P|$ the number of edges of $P$.

\begin{lemma}[Bellman--Ford algorithm~\cite{bellman1958routing}]\label{L-bellman1958routing}
    Given a weighted undirected graph $G$ and a source vertex $u$, $k$ scans of all edges compute $\delta_G^{\le k}(u, v)$ for every $v\in V$ in $O(mk)$ time.
    If the labels $d[\cdot]$ are updated in place, then after $i$ scans, $\delta_G(u,v)\le d[v]\le\delta_G^{\le i}(u,v)$ for every $v\in V$, and in particular $d[v]=\delta_G(u,v)$ whenever $i\ge|P_G(u,v)|$.
\end{lemma}
The in-place bound holds since every finite label is the weight of a walk, and using labels that were already updated in the current scan can only decrease the values.

\subsection{Balls and Bundles}\label{S-balls-bundles}
We define bundles and balls as follows. Let $S \subseteq V$ be a non-empty set.
For each $v \in V$, let $b(v)=\arg\min_{s \in S} \delta(v,s)$ be the closest vertex in $S$ to $v$, where the minimum is taken with respect to the total order $\prec$ of \Cref{S-positive-weights}; by item~\ref{P-sp-distinct} of \Cref{P-shortest-paths} it is attained by a unique $s\in S$. The path from $s$ to itself has the smallest key $(0,0)$, so $b(s)=s$ for every $s\in S$.
We say that $v$ is \emph{bundled} to $b(v)$.

For each $s\in S$, the \emph{bundle} of $s$ is $\Bundle(s)=\{v\in V : b(v)=s\}$; notice that $\{\Bundle(s)\}_{s\in S}$ partitions $V$.
Finally, the \emph{ball} of $v$ consists of $b(v)$ together with all vertices closer to $v$ than $b(v)$, that is, $\Ball(v)=\{w\in V : \delta(v,w)\prec\delta(v,b(v))\} \cup \{b(v)\}$, where, as everywhere in this paper, distances from the common vertex $v$ are compared under $\prec$; note that $\Ball(s)=\{s\}$ for every $s\in S$. See \Cref{fig:bundle-ball} for an illustration of the definitions.

\colorlet{cbund}{blue!65!black}    
\colorlet{cballA}{violet}          
\colorlet{cballB}{orange!90!black} 

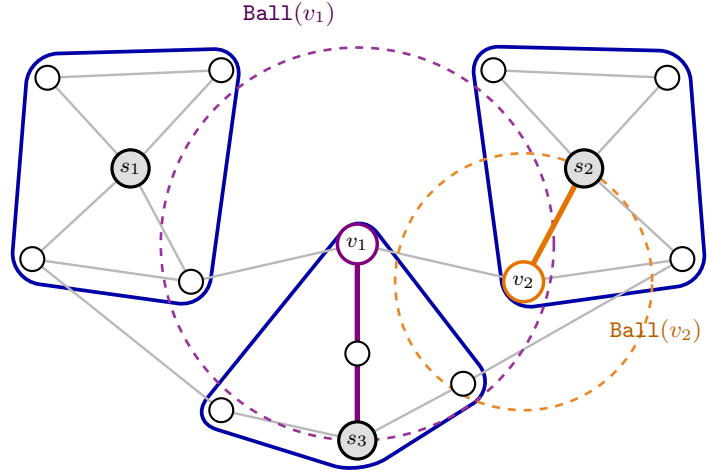
\begin{figure}[t]
\centering
\begin{tikzpicture}[
    vert/.style    = {circle, draw=black, fill=white, minimum size=9pt, inner sep=0pt, line width=0.8pt, font=\scriptsize},
    svert/.style   = {circle, draw=black, fill=gray!25, minimum size=14pt, inner sep=0pt, line width=1.2pt, font=\scriptsize},
    centerA/.style = {circle, draw=cballA, fill=white, minimum size=15pt, inner sep=0pt, line width=1.3pt, font=\scriptsize},
    centerB/.style = {circle, draw=cballB, fill=white, minimum size=15pt, inner sep=0pt, line width=1.3pt, font=\scriptsize},
    gedge/.style   = {gray!55, line width=0.9pt},
]

\begin{scope}[rounded corners=11pt, line width=1.4pt, line join=round, cbund]
    \draw (-1.59,0.52) -- (1.01,0.16) -- (1.48,3.59) -- (-1.35,3.51) -- cycle;
    \draw ( 4.96,0.12) -- (7.62,0.49) -- (7.38,3.55) -- ( 4.49,3.62) -- cycle;
    \draw ( 3.03,1.45) -- (4.85,-0.91) -- (3.05,-2.05) -- (0.77,-1.34) -- cycle;
\end{scope}

\draw[gedge] (0,2)--(-1.1,3.2);   \draw[gedge] (0,2)--(-1.3,0.8);
\draw[gedge] (0,2)--(1.2,3.3);    \draw[gedge] (0,2)--(0.8,0.5);
\draw[gedge] (-1.1,3.2)--(1.2,3.3);  \draw[gedge] (-1.3,0.8)--(0.8,0.5);
\draw[gedge] (6,2)--(7.1,3.2);    \draw[gedge] (6,2)--(7.3,0.8);
\draw[gedge] (6,2)--(4.8,3.3);    \draw[gedge] (6,2)--(5.2,0.5);
\draw[gedge] (7.1,3.2)--(4.8,3.3);   \draw[gedge] (7.3,0.8)--(5.2,0.5);
\draw[gedge] (3,-1.6)--(3,-0.45); \draw[gedge] (3,-1.6)--(4.4,-0.85);
\draw[gedge] (3,-1.6)--(1.2,-1.2);   \draw[gedge] (3,-0.45)--(3,1.0);
\draw[gedge] (3,1.0)--(0.8,0.5);  \draw[gedge] (3,1.0)--(5.2,0.5);
\draw[gedge] (-1.3,0.8)--(1.2,-1.2); \draw[gedge] (7.3,0.8)--(4.4,-0.85);

\draw[cballA!80, dashed, line width=1pt] (3,1.0) circle (2.6);
\draw[cballA, line width=2pt] (3,1.0)--(3,-0.45)--(3,-1.6);     
\draw[cballB!85, dashed, line width=1pt] (5.2,0.5) circle (1.70);
\draw[cballB, line width=2pt] (5.2,0.5)--(6,2);                 

\node[svert] (s1) at (0,2)    {$s_1$};
\node[svert] (s2) at (6,2)    {$s_2$};
\node[svert] (s3) at (3,-1.6) {$s_3$};
\node[vert] at (-1.1,3.2){}; \node[vert] at (-1.3,0.8){};
\node[vert] at (1.2,3.3) {}; \node[vert] at (0.8,0.5) {};
\node[vert] at (7.1,3.2) {}; \node[vert] at (7.3,0.8) {};
\node[vert] at (4.8,3.3) {};
\node[vert] at (3,-0.45) {}; \node[vert] at (4.4,-0.85){};
\node[vert] at (1.2,-1.2){};
\node[centerA] at (3,1.0)   {$v_1$};
\node[centerB] at (5.2,0.5) {$v_2$};

\node[font=\footnotesize, cballA!65!black]          at (2.1,4.05)  {$\Ball(v_1)$};
\node[font=\footnotesize, cballB!80!black]          at (6.95,-0.15){$\Ball(v_2)$};

\end{tikzpicture}
\caption{%
    Bundles (blue) partition $V$ by nearest sample $s_i\in S$. Balls (coloured)
    contain $b(v)$ and all vertices closer to $v$ than $b(v)$. They may cross
    several bundles and overlap one another.%
}
\label{fig:bundle-ball}
\end{figure}

Let $G'=(V\cup\{r\},E\cup\{(r,s,0)\mid s\in S\})$ be the graph obtained from $G$ by adding a super-source $r$ connected to every $s\in S$ by a weight-$0$ edge.
The edges $(r,s)$ are not perturbed, so the key of a path of $G'$ that starts with an edge $(r,s)$ is the key of its subpath from $s$.
Since a path and its reverse have the same key, the definition of $b(v)$ implies that $P_{G'}(r,v)$ is the edge $(r,b(v))$ followed by $P_G(b(v),v)$, the reverse of $P_G(v,b(v))$ (item~\ref{P-sp-reverse} of \Cref{P-shortest-paths}).
In particular, $\delta_{G'}(r,v)=\delta_G(b(v),v)$ for every $v\in V$.

Both \BoundedBellmanFord and \ActiveBellmanFord also maintain, alongside $d[\cdot]$, a label $b[\cdot]$ that is overwritten whenever $d[\cdot]$ is; the following lemma is what we use about it, and it too is stated for in-place execution.

\begin{lemma}[{Correctness of the labels $b[\cdot]$}]\label{L-labels}
Consider an execution of \Cref{alg:bounded-bf} or of \iflncs\else\Cref{alg:active-bf}\fi{} on $G$ and $S\neq\emptyset$, and let $G'=(V\cup\{r\},E\cup\{(r,s,0)\mid s\in S\})$ be the graph with the auxiliary super-source $r$. At every point of the execution and for every $v\in V$, if $d[v]=\delta_{G'}(r,v)$ then $b[v]=b(v)$.
\end{lemma}
\begin{proof}
Recall that a label $d[v]$ always represents a concrete walk from $r$ to $v$ in $G'$, and that labels are compared under the total order $\prec$ of \Cref{S-positive-weights}; thus $d[v]=\delta_{G'}(r,v)$ means that the walk represented by $d[v]$ is exactly $P_{G'}(r,v)$.
We argue by induction on $|P_{G'}(r,v)|$.
If $d[v]$ still holds its initial value, then $v\in S$ and $b[v]=v$; moreover $P_{G'}(r,v)$ is then the single edge $(r,v)$, so $b(v)=v$, as required.
Otherwise, consider the last relaxation that wrote $d[v]$, say along an edge $(x,v)$, which also set $b[v]\leftarrow b[x]$; the walk represented by $d[v]$ from that moment on is the walk represented by $d[x]$ at that moment, extended by $(x,v)$.
Since $d[v]=\delta_{G'}(r,v)$, that walk is $P_{G'}(r,v)$; hence $x$ is the predecessor of $v$ on $P_{G'}(r,v)$ and $d[x]$ represented the prefix $P_{G'}(r,x)$, i.e., $d[x]=\delta_{G'}(r,x)$ at that moment (item~\ref{P-sp-tree} of \Cref{P-shortest-paths}).
As $|P_{G'}(r,x)|=|P_{G'}(r,v)|-1$, the induction hypothesis gives $b[x]=b(x)$ at that moment, so $b[v]=b(x)$.
Finally, $P_{G'}(r,x)$ is a prefix of $P_{G'}(r,v)$, so the two paths share their first edge, which is $(r,b(v))$; hence $b(x)=b(v)$ and $b[v]=b(v)$.
Since $\delta_{G'}(r,v)$ is the $\prec$-minimum, $d[v]$ is never written again, so the equality persists.
\end{proof}

\section{Faster Algorithm for SSSP in Weighted Undirected Graphs}\label{S-small-improve}
In this section we present the main result of this paper, an algorithm that is faster than the state-of-the-art algorithm of~\cite{duan2023randomized} for single-source shortest paths in sparse weighted undirected graphs.
\theoremImprovedUndirected*

We first review the algorithm of \cite{duan2023randomized} and then describe our improvement.

\subsection{The Framework of~\cite{duan2023randomized}}
Let $S$ be a random set formed by sampling each vertex independently with probability $1/k$ and adding the source $u$, as in~\cite{duan2023randomized}. Note that $\EE[|S|]\le n/k+1$ and $\EE[|\Ball(v)|]\le k$ for every $v\in V$.

For simplicity of presentation, assume first that we are in a \emph{balanced case}, in which $|S|=O(n/k)$ and $|\Ball(v)|\le ck$ for every $v\in V$, for a constant $c>0$.\footnote{The balanced case need not hold when $S$ is randomly sampled; it is only used for clarity of presentation. In \FullRef{S-Active-BF} we remove this assumption.}

The $O(m\log^{1/2}{n}(\log\log{n})^{1/2})$ time algorithm of \cite{duan2023randomized} works as follows. 
First, it computes the distances $\delta(v,b(v))$ for every $v\in V$, the bundles $\Bundle(s)$ for every $s\in S$, and the balls $\Ball(v)$ for every $v\in V$.
Then, it invokes \BundleDijkstra (\Cref{L-BundleDijkstra}) from the source $u$ to complete the SSSP computation.
\begin{lemma}[\BundleDijkstra~\cite{duan2023randomized}]\label{L-BundleDijkstra}
    Given a weighted undirected graph $G$, a source $u\in V$, a set $S\subseteq V$ with $u\in S$ and $|S|=O(n/k)$, and $\delta(v,w)$ for every $v\in V$ and $w\in\Ball(v)$, algorithm \BundleDijkstra computes the single-source shortest paths from $u$ to every $v\in V$ in
    $O\!\left(\frac{m}{n}\sum_{v\in V}|\Ball(v)| + \frac{n}{k}\log n\right)$ time, where $m/n$ bounds the maximum degree of $G$.
\end{lemma}

Notice that in the balanced case, given the bundles and balls, the algorithm \BundleDijkstra completes the SSSP computation in $O(mk+(n/k)\log{n})$ time, since $|\Ball(v)|\le ck$ for every $v\in V$.
Perhaps surprisingly, the dominant cost of the algorithm of \cite{duan2023randomized} is not \BundleDijkstra but the construction of bundles and balls.
To compute bundles and balls, \cite{duan2023randomized} runs a separate bounded Dijkstra from each $v\in V$, stopping when the first vertex of $S$ is reached, thus computing $b(v)$ and $\Ball(v)$ in $O(k\log k)$ time per vertex and $O(mk\log k)$ time overall, which dominates the $O(mk+(n/k)\log n)$ cost of \BundleDijkstra. 
Therefore, in this paper we focus on efficient computation of bundles and balls.

\subsection{Speeding Up the Ball and Bundle Construction}\label{S-construction}
First, we present an algorithm that, under the balanced case assumption, computes the distances $\delta(v,b(v))$ and the bundles in $O(mk)$ time, improving the $O(mk\log{k})$ running time of \cite{duan2023randomized}.

The algorithm, \BoundedBellmanFord (\Cref{alg:bounded-bf}), runs a fixed number $t$ of Bellman--Ford steps from an auxiliary super-source connected to all samples, maintaining for each vertex a tentative distance $d[v]$ together with a label $b[v]$ recording the sample that realizes it; the label is carried along whenever $d[v]$ is relaxed.
Taking $t=\ceil{ck+1}$, we show that under the balanced case assumption, this suffices to compute $\delta(v,b(v))$ and the bundles.

\begin{algorithm2e}[t]
\caption{\BoundedBellmanFord$(G,S,t)$}\label{alg:bounded-bf}
\DontPrintSemicolon
\lForEach{$v\in V$}{$d[v]\leftarrow\infty$}
\ForEach{$s\in S$}{$d[s]\leftarrow 0$\; $b[s]\leftarrow s$}
\For{$i\leftarrow 1$ \KwTo $t$}{
    \ForEach{$(v,w)\in E$}{
        \If{$d[v]+\wt(v,w)<d[w]$}{
            $d[w]\leftarrow d[v]+\wt(v,w)$\;
            $b[w]\leftarrow b[v]$\;
        }
    }
}
\Return $\{d[v]\}_{v\in V}$, $\{b[v]\}_{v\in V}$\;
\end{algorithm2e}

\colorlet{csample}{blue!65!black}        
\colorlet{csource}{red!65!black}         
\colorlet{cpath}{violet}                 
\colorlet{cpathB}{orange!85!black}       

\begin{figure}[t]
\centering
\begin{tikzpicture}[
    vert/.style   = {circle, draw=black,   fill=white,     minimum size=9pt,  inner sep=0pt, line width=0.8pt, font=\scriptsize},
    svert/.style  = {circle, draw=csample, fill=blue!12,   minimum size=14pt, inner sep=0pt, line width=1.2pt, font=\scriptsize},
    rvert/.style  = {circle, draw=csource, fill=red!10,    minimum size=15pt, inner sep=0pt, line width=1.5pt, font=\footnotesize},
    pvert/.style  = {circle, draw=cpath,   fill=violet!8,  minimum size=9pt,  inner sep=0pt, line width=1.0pt, font=\scriptsize},
    pvertB/.style = {circle, draw=cpathB,  fill=orange!8,  minimum size=9pt,  inner sep=0pt, line width=1.0pt, font=\scriptsize},
    gedge/.style  = {gray!55, line width=0.9pt},
    zedge/.style  = {csource, dashed, line width=1.1pt},
    pedge/.style  = {cpath,  line width=2pt, line cap=round},
    pedgeB/.style = {cpathB, line width=2pt, line cap=round},
]

\draw[rounded corners=10pt, gray!45, line width=1pt, dashed]
    (-1.5,-2.2) rectangle (9.0,3.2);
\node[font=\small, gray!50!black] at (-1.0, 2.85) {$G$};

\draw[gedge] (0,1.7)    -- (-0.8,0.3);
\draw[gedge] (0,1.7)    -- (0.8,-1.7);
\draw[gedge] (-0.8,0.3) -- (0.8,-1.7);
\draw[gedge] (1.5,-0.6) -- (0.8,-1.7);
\draw[gedge] (0,1.7)    -- (2.6,0.9);
\draw[gedge] (2.6,0.9)  -- (1.5,-0.6);    
\draw[gedge] (3.5,2.2)  -- (2.6,0.9);     
\draw[gedge] (3.5,2.2)  -- (5.2,2.8);
\draw[gedge] (3.5,2.2)  -- (4.8,0.3);
\draw[gedge] (5.2,2.8)  -- (7.2,1.1);
\draw[gedge] (4.8,0.3)  -- (7.2,1.1);
\draw[gedge] (7.2,1.1)  -- (8.2,0.2);
\draw[gedge] (4.8,0.3)  -- (5.8,-1.0);
\draw[gedge] (5.8,-1.0) -- (8.2,0.2);

\draw[pedge]  (3.5,2.2) -- (2.6,0.9) -- (1.5,-0.6);   
\draw[pedgeB] (0,1.7)   -- (-0.8,0.3) -- (0.8,-1.7);  

\node[svert]  (s1) at (0,1.7)    {$s_1$};
\node[svert]  (s2) at (3.5,2.2)  {$s_2$};
\node[svert]  (s3) at (7.2,1.1)  {$s_3$};
\node[pvertB] at (-0.8,0.3)      {};
\node[pvertB] (v1) at (0.8,-1.7) {$v_1$};
\node[pvert]  at (2.6,0.9)       {};
\node[pvert]  (v2) at (1.5,-0.6) {$v_2$};
\node[vert]   at (5.2,2.8)       {};
\node[vert]   at (4.8,0.3)       {};
\node[vert]   at (5.8,-1.0)      {};
\node[vert]   at (8.2,0.2)       {};

\node[font=\scriptsize, cpathB,  anchor=east] at (-0.2,1.88)  {$b(v_1)\!=\!s_1$};
\node[font=\scriptsize, csample, anchor=west] at (3.7,1.88)   {$b(v_2)\!=\!s_2$};
\node[font=\scriptsize, cpathB,  anchor=east] at (-0.3,-0.5)  {$\delta_G(v_1,b(v_1))$};
\node[font=\scriptsize, cpath,   anchor=east] at (2.35,0.70)  {$\delta_G(v_2,b(v_2))$};

\node[rvert] (r) at (3.5,5.6) {$r$};

\draw[zedge] (r) -- (s1)
    node[pos=0.53, left,  font=\scriptsize, csource] {$0$};
\draw[zedge] (r) -- (s2)
    node[pos=0.53, right, font=\scriptsize, csource, xshift=2pt] {$0$};
\draw[zedge] (r) -- (s3)
    node[pos=0.53, right, font=\scriptsize, csource] {$0$};

\end{tikzpicture}
\caption{%
    The super-source construction of \Cref{L-BF-good-case}.
    Weight-$0$ edges (red) connect $r$ to every $s\in S$, giving
    $\delta_{G'}(r,v)=\delta_G(v,b(v))$ for every $v$
    (violet: $b(v_2)=s_2$; orange: $b(v_1)=s_1$).
}
\label{fig:super-source}
\end{figure}
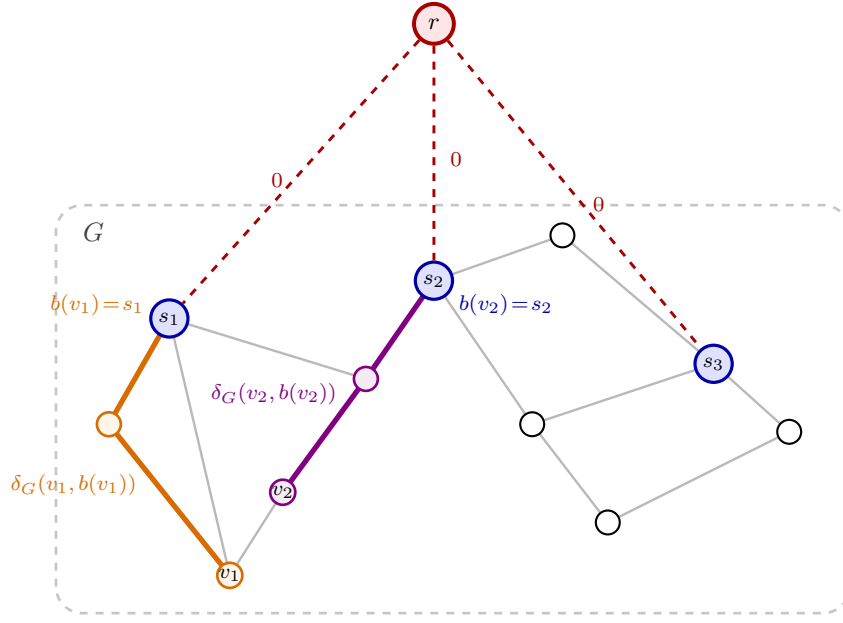

\begin{lemma}[Computing bundles efficiently in the balanced case]\label{L-BF-good-case}
    Let $c>0$ be a constant, and let $S$ be a set such that $|\Ball(v)|\le ck$ for every $v\in V$. Then $\BoundedBellmanFord(G,S,\ceil{ck+1})$ (\Cref{alg:bounded-bf}) computes $d[v]=\delta(v,b(v))$ for every $v\in V$ and, after grouping vertices by $b[\cdot]$, the bundles $\Bundle(s)$ for every $s\in S$, in $O(mk)$ time.
\end{lemma}
\begin{proof}
The execution of \Cref{alg:bounded-bf} coincides with running $t=\ceil{ck+1}$ steps of Bellman--Ford from an auxiliary super-source $r$ in the graph $G'=(V\cup\{r\},E'=E\cup \{(r,s,0)\mid s\in S\})$, where $r$ is connected to every $s\in S$ by a weight-$0$ edge: the initialization $d[s]=0$ for $s\in S$ corresponds to relaxing the edges $(r,s)$, and each subsequent iteration is one Bellman--Ford step over $E$.
As shown in \Cref{S-balls-bundles}, $\delta_{G'}(r,v)=\delta_G(b(v),v)$ for every $v\in V$ (see \Cref{fig:super-source} for an illustration).

The key insight is that $P_G(v,b(v))\subseteq\Ball(v)$:
every vertex on the path to $b(v)$ lies strictly closer to $v$ than $b(v)$ itself, placing it inside the ball, so $|P_{G'}(r,v)|\leq |\Ball(v)|+1$, and therefore $\delta_{G'}^{\leq |\Ball(v)|+1}(r,v)=\delta_{G'}(r,v)$.

Formally, for every $x\in P_G(v,b(v))$ with $x\neq b(v)$, the subpath $P_G(v,x)$ is a proper prefix of $P_G(v,b(v))$ and is therefore strictly smaller under the total order of \Cref{S-positive-weights} (item~\ref{P-sp-prefix} of \Cref{P-shortest-paths}), so $\delta(v,x)\prec\delta(v,b(v))$, meaning $x\in\Ball(v)$. Note that this holds even when the two subpaths have the same weight, as is the case when the subpath from $x$ to $b(v)$ consists of weight-$0$ edges.

By definition of $\Ball(v)$, we have $b(v)\in\Ball(v)$ and therefore $P_G(v,b(v))\subseteq\Ball(v)$. This implies that $|P_G(v,b(v))|\le|\Ball(v)|\le ck$ and $|P_{G'}(r,v)|=1+|P_G(v,b(v))|\le ck+1=t$.
Therefore, we get: $$\delta_{G'}^{\le t}(r,v) = \delta_{G'}(r,v)=\delta(v,b(v)).$$

Thus, by \Cref{L-bellman1958routing}, after the $t=\ceil{ck+1}$ steps, $d[v]=\delta_{G'}(r,v)=\delta_G(v,b(v))$ for every $v\in V$.
By \Cref{L-labels}, at that point $b[v]=b(v)$ for every $v\in V$; grouping vertices by $b[v]$ thus yields $\Bundle(s)$ for every $s\in S$.
The $t=O(k)$ Bellman--Ford steps take $O(mk)$ time by \Cref{L-bellman1958routing}, which dominates the $O(n)$ grouping step.
\end{proof}

Having computed the distances $\delta(v,b(v))$ and the bundles, it remains to construct $\Ball(v)$ for every $v\in V$.
We construct the balls using a bounded-distance search $\BoundSSSP(v,B)$, which computes $\delta(v,w)$ for every $w$ in $U=\{w\mid\delta(v,w)\prec B\}$, and which we obtain from the recursive procedure $\BMSSP$ of~\cite{duan2026faster}.

This requires some care, for two reasons.
First, a call to $\BMSSP$ is allowed to end in a \emph{partial execution}: it may return a bound $B'\prec B$ and complete only the vertices at distance smaller than $B'$, in which case it does not compute all of $U$.
Second, the running time of $\BMSSP$ depends on two parameters, $t$ and $\ell$, whose optimal setting depends on $|U|$ --- which is not known in advance, as $|U|=O(|\Ball(v)|)$ varies over the sources $v$.
In \FullRef{A-BoundSSSP} we resolve both issues, by invoking $\BMSSP$ with a doubling parameter $t$ and an increasing level $\ell$ until a full execution is reached. By geometric summation, the total cost of the unsuccessful calls is bounded by the same asymptotic running-time bound as the successful phase. This yields the following lemma.
\begin{lemma}[{\BoundSSSP; derived from Lemma~3.1 of \cite{duan2026faster} in \FullRef{A-BoundSSSP}}]\label{L-Faster-D-SSSP}
    Let $G$ be a weighted directed graph with maximum degree $\Delta\ge 1$, let $v\in V$, let $B$ be a bound, and let $U=\{w\mid \delta(v,w)\prec B\}$.
    After a single global initialization taking $O(n)$ time, there is an algorithm $\BoundSSSP(v,B)$ that computes $\delta(v,w)$ for every $w\in U$ in $O(\Delta\,|U|\log^{1/2}{|U|}(\log\log|U|)^{1/2})$ time.
    In particular, on a graph preprocessed by the degree reduction of \Cref{S-degree-reduction}, in which $\Delta=m/n$, the running time is $O(\frac{m}{n}|U|\log^{1/2}{|U|}(\log\log|U|)^{1/2})$.
\end{lemma}

Recall from \Cref{L-BF-good-case} that $b(v)$ and the distance $\delta(v,b(v))$ are already known for every $v\in V$, so we can call $\BoundSSSP(v,B)$ for each $v$ with $B$ set to $\delta(v,b(v))$.
The set $U=\{w \mid \delta(v,w)\prec\delta(v,b(v))\}$ returned by $\BoundSSSP(v,B)$ is exactly $\Ball(v)\setminus\{b(v)\}$, so adding the already-known $b(v)$ yields $\Ball(v)$. In the balanced case we have $|\Ball(v)|= O(k)$, so the cost per vertex is $O(\frac{m}{n}|\Ball(v)|\log^{1/2}{|\Ball(v)|}(\log\log|\Ball(v)|)^{1/2})=O(\frac{m}{n}k\log^{1/2}{k}(\log\log{k})^{1/2})$.

We are now ready to present the complete algorithm under the balanced case assumption. The algorithm $\ImprovedSSSPBalanced$ is composed of three steps.
\begin{enumerate}
    \item\label{step:bundles} Run $\BoundedBellmanFord(G,S,\ceil{ck+1})$ (\Cref{alg:bounded-bf}) to compute $\delta(v,b(v))$ for every $v\in V$ and the bundles $\Bundle(s)$ for every $s\in S$ (\Cref{L-BF-good-case}).
    \item\label{step:balls} Run $\BoundSSSP(v,B)$ for every $v\in V$ with bound $B=\delta(v,b(v))$ to obtain $\Ball(v)$.
    \item\label{step:bundle-dijkstra} Call \BundleDijkstra of \cite{duan2023randomized} with the bundles computed in step~\ref{step:bundles} and the balls computed in step~\ref{step:balls}, and return the resulting shortest-path distances.
\end{enumerate}
\begin{lemma}[\Cref{T-Improved-undirected} in the balanced case]\label{L-good-case}
    Let $c>0$ be a constant, and let $S$ be a set with $u\in S$ such that $|S|=O(n/k)$ and $|\Ball(v)|\le ck$ for every $v\in V$. Then algorithm \ImprovedSSSPBalanced computes SSSP in weighted undirected graphs in \[O\!\left(mk\log^{1/2}{k}(\log\log{k})^{1/2}+\frac{n}{k}\log{n}\right)\] time. 
    Setting $k$ to $\log^{1/2}{n}\cdot(\log\log{n})^{-1/4}\cdot(\log\log\log{n})^{-1/4}$, we get \[O\!\left(m\log^{1/2}{n}(\log\log{n})^{1/4}(\log\log\log{n})^{1/4}\right)\] time.
\end{lemma}
\iflncs
\else
We prove \Cref{L-good-case} in \Cref{A-balanced-proof}, and remove the balanced case assumption in \Cref{S-Active-BF}.
\fi

\section{Discussion and Open Problems}\label{S-discussion}

\paragraph{General $m$.}
For general $m$, the degree reduction of \Cref{S-degree-reduction} with $\Delta=m/n$ yields a graph with $O(n)$ vertices and maximum degree $m/n$, as in~\cite[Section~V]{duan2023randomized}.
On this graph, the bundle construction of~\cite{duan2023randomized} takes $O(mk+nk\log(mk/n))$ time, and \BundleDijkstra takes $O(mk+\frac{n}{k}\log n)$ time (\Cref{L-BundleDijkstra}).
Our algorithm computes the bundles in $O(mk)$ time (\Cref{T-active-BF}).
It computes the balls in $O(mk+nk\sqrt{(m/n)\log k\log\log k})$ time. This follows from two bounds of \iflncs\else\Cref{A-BoundSSSP,A-sampling}\fi: the bound $O(\Delta|U|+|U|\sqrt{\Delta\log|U|\log\log|U|})$ on the running time of \BoundSSSP, and the bound on the weighted sum of ball sizes of a random sample, whose proof also holds when $m/n\le\log n$.
Hence our construction is faster than that of~\cite{duan2023randomized} when $m/n=o(\log k/\log\log k)$.
For larger $m$, the balls can be constructed as in~\cite{duan2023randomized}, so the running time is never worse than that of~\cite{duan2023randomized}.

\paragraph{The ball construction bottleneck.}
\Cref{T-active-BF} computes the distances $\delta(v,b(v))$ and the bundles in $O(mk)$ time, so the running time of our algorithm is now dominated by the construction of the balls, which takes $O(mk\log^{1/2}{k}(\log\log{k})^{1/2})$ time using \BoundSSSP of~\cite{duan2026faster} (\Cref{L-Faster-D-SSSP}).
Any improvement to the ball construction immediately improves the overall running time.
In particular, an $O(mk)$ time construction of all the balls would, by setting $k=\log^{1/2}{n}$, yield an $O(m\log^{1/2}{n})$ time algorithm for undirected SSSP.
Whether the balls can be constructed in $O(mk)$ time, and, more generally, what the true complexity of SSSP in the comparison-addition model is, are intriguing open problems.

\paragraph{Directed graphs.}
The active-vertex Bellman--Ford primitive itself extends to directed graphs.
Both Bellman--Ford and the super-source reduction $\delta_{G'}(r,v)=\delta(v,b(v))$ are direction-agnostic, and the amortization goes through when the analysis is carried out with respect to in-balls,\footnote{In a directed graph the ball splits into the \emph{out-ball} $\{w\mid\delta(v,w)<B\}$, using the distances out of $v$, and the \emph{in-ball} $\{w\mid\delta(w,v)<B\}$, using the distances into $v$; the two coincide when $G$ is undirected.} the balls that govern the super-source computation.
The obstruction to a directed result lies instead in \BundleDijkstra (\Cref{L-BundleDijkstra}) of~\cite{duan2023randomized}, which relies on the graph being undirected so that the in-ball and out-ball of each vertex coincide.
We leave as an open problem the adaptation of \BundleDijkstra to directed graphs, which would make our improvement carry over to directed SSSP as well.

\paragraph{Derandomization.}
Given any $S$ satisfying the three properties of a random sample (see \FullRef{L-sampling}), the entire algorithm is deterministic, so derandomization reduces to constructing a single such $S$.
A deterministic construction must certify that $\sum_{v}|\Ball(v)|=O(nk)$, and the natural certificate, growing the balls, costs $O(mk\log{k})$, as in the derandomization of~\cite{yan2026lossless}.
We leave open whether such an $S$ can be constructed deterministically in $o(mk\log{k})$ time, which would give an improved deterministic algorithm.

\iflncs
\else
\clearpage
\bibliographystyle{alpha}
\fi
\bibliography{articles}

\iflncs\else
\clearpage
\appendix
\section*{Appendix}
\section{Proof of the Balanced-Case Running Time}\label{A-balanced-proof}
\begin{proof}
    Correctness follows directly from \Cref{L-BF-good-case} (bundles and $b(v)$), \Cref{L-Faster-D-SSSP} (balls), and \Cref{L-BundleDijkstra} (\BundleDijkstra).
    The running time of the algorithm is:
    \begin{enumerate}
        \item Computing $\Bundle(s)$ for every $s\in S$, and $b(v)$ for every $v\in V$, takes $O(mk)$ time (\Cref{L-BF-good-case}).
        \item Computing $\Ball(v)$ for every $v\in V$ by $n$ calls to \BoundSSSP with $B=\delta(v,b(v))$, each on a set of size $|\Ball(v)|\le ck$, takes time $O(mk\log^{1/2}{k}(\log\log{k})^{1/2})$.
        \item Computing SSSP by calling \BundleDijkstra (\Cref{L-BundleDijkstra}) takes $O(mk + (n/k)\log{n})$ time.
    \end{enumerate}

    Combining the three steps, the total running time is
    $$O\!\left(mk\log^{1/2}{k}(\log\log{k})^{1/2}+\frac{n}{k}\log{n}\right).$$
    Setting $k=\log^{1/2}{n}\cdot(\log\log{n})^{-1/4}\cdot(\log\log\log{n})^{-1/4}$ balances the terms and yields a running time of
    $$O\!\left(m\log^{1/2}{n}(\log\log{n})^{1/4}(\log\log\log{n})^{1/4}\right).$$
\end{proof}

\section{Computing $\delta(v,b(v))$ for Every $v\in V$ in the General Case}\label{S-Active-BF}

In this section, we remove the balanced case assumption and extend our analysis to an arbitrary randomly sampled set $S$. Recall that \Cref{L-good-case} assumes that $|\Ball(v)|\le ck$ for every $v\in V$. Under random sampling, however, this assumption cannot be guaranteed to hold simultaneously for all vertices.
For example, in a path graph with unit weights, w.h.p.\ there are two sampled vertices $s_1,s_2\in S$ such that $\delta(s_1,s_2)=\Omega(k\log n)$ and there is no vertex from $S$ on the path between $s_1$ and $s_2$. In such a case, vertices near the middle of $P(s_1,s_2)$ have balls of size $\Omega(k\log n)$, far exceeding the $O(k)$ bound of the balanced case.

What holds with high probability is only the weaker global bound $\sum_{v\in V}|\Ball(v)|=O(nk)$. 
Our
algorithm uses this weaker bound together with two additional properties of a random sample, all stated in the following lemma:

\begin{restatable}{lemma}{lemmaSampling}\label{L-sampling}
    Let $\log^{0.25}{n}\le k\le \log^{0.99}{n}$, and let $S$ be a set in which each vertex of $V$ is picked independently with probability $1/k$. Then the following properties hold with high probability:
    \begin{itemize}
        \item $|S|=O(n/k)$.
        \item $\sum_{v\in V}|\Ball(v)| = O(nk)$.
        \item The weighted sum satisfies
        \[\begin{aligned}
        &\sum_{v\in V}|\Ball(v)|\cdot \log^{1/2}|\Ball(v)|\cdot(\log\log|\Ball(v)|)^{1/2}\\
        &\qquad = O(nk\log^{1/2}{k}(\log\log{k})^{1/2}).
        \end{aligned}\]
    \end{itemize}
\end{restatable}
\begin{proof}
The proof of this lemma appears in \Cref{A-sampling}.    
\end{proof}

We modify the \ImprovedSSSPBalanced algorithm so that it no longer relies on the balanced case assumption $|\Ball(v)| = O(k)$ for every $v \in V$.
Instead, the modified algorithm (\ImprovedSSSP) works with a randomly sampled set $S$, for which only the weaker global bound $\sum_{v \in V} |\Ball(v)| = O(nk)$ is guaranteed, as proved in \Cref{L-sampling}.

Steps~\ref{step:balls} and~\ref{step:bundle-dijkstra} of the \ImprovedSSSPBalanced algorithm remain unchanged for a random sample $S$ and retain the same running time.
In step~\ref{step:balls}, the total cost of calling \BoundSSSP over all vertices is \[\sum_{v \in V} O\!\left(\frac{m}{n}|\Ball(v)| \log^{1/2} |\Ball(v)| (\log\log|\Ball(v)|)^{1/2}\right).\] By \Cref{L-sampling}, this sum is bounded by $O\!\left(mk \log^{1/2} k (\log\log{k})^{1/2}\right)$, as in the balanced case.
In step~\ref{step:bundle-dijkstra}, we call \BundleDijkstra. Its running time is $O\!\left(\frac{m}{n}\sum_{v \in V} |\Ball(v)| + \frac{n}{k}\log n\right)$, provided that $|S| = O(n/k)$. By \Cref{L-sampling}, we have both $|S| = O(n/k)$ and $\sum_{v \in V} |\Ball(v)| = O(nk)$. Therefore, \BundleDijkstra runs in $O\!\left(mk + \frac{n}{k}\log n\right)$ time, as in the balanced case.

It remains to address step~\ref{step:bundles}, namely, how to compute $\delta(v,b(v))$ for every $v \in V$ and construct the bundles in $O(mk)$ time for a randomly sampled set $S$, without assuming that $|\Ball(v)| \le ck$ for every $v \in V$.
The algorithm \BoundedBellmanFord (\Cref{alg:bounded-bf}) is no longer applicable, since it terminates after a fixed number, $\ceil{ck+1}$, of Bellman--Ford iterations. This stopping criterion relies on the assumption that $|\Ball(v)| \le ck$ and on prior knowledge of the constant $c$. 

To overcome this difficulty, we use the active-vertex (queue-based) Bellman--Ford algorithm~\cite{moore1959shortest} from an auxiliary super-source.
We prove that it computes $\delta(v,b(v))$ for every $v\in V$ in $O\!\left(\sum_{v\in V} |\Ball(v)| \cdot m/n\right)$ time (\Cref{T-active-BF}), which, combined with the bound $\sum_{v\in V} |\Ball(v)| = O(nk)$ (\Cref{L-sampling}), yields an $O(mk)$ running time (\Cref{C-active-BF-whp}).

The \ActiveBellmanFord algorithm (\Cref{alg:active-bf}) takes as input a graph $G$ and a set $S \subseteq V$, and maintains three structures: for each vertex $v \in V$, a tentative distance $d[v] \in \mathbb{R}_{\ge 0} \cup \{\infty\}$ to the nearest sample, a tentative nearest sample $b[v] \in S$, and a set $\mathit{Active} \subseteq V$ of vertices whose tentative distance decreased in the previous round of the algorithm.
The algorithm initializes $d[s]=0$ and $b[s]=s$ for every $s\in S$, $d[v]=\infty$ for every $v\notin S$, and $\mathit{Active}=S$.
The algorithm then proceeds in rounds until $\mathit{Active}=\emptyset$.
In each round, for every vertex $v\in\mathit{Active}$ the algorithm relaxes each incident edge $(v,w)\in E$: if $d[v]+\wt(v,w) < d[w]$, the algorithm sets $d[w]\leftarrow d[v]+\wt(v,w)$ and $b[w]\leftarrow b[v]$.
The set of vertices whose tentative distance decreased during the round becomes the set $\mathit{Active}$ of the next round.
When $\mathit{Active}=\emptyset$ the algorithm returns $\{d[v]\}_{v\in V}$ and $\{b[v]\}_{v\in V}$.
See \Cref{alg:active-bf} for pseudocode and \Cref{fig:active-bf-rounds} for an illustration of its execution.
\colorlet{csample}{blue!65!black}      

\begin{figure}[t]
\centering
\begin{tikzpicture}[
    vert/.style  = {circle, draw=black,   fill=white,   minimum size=15pt, inner sep=0pt, line width=0.8pt, font=\scriptsize},
    svert/.style = {circle, draw=csample, fill=blue!12, minimum size=15pt, inner sep=0pt, line width=1.2pt, font=\scriptsize, text=csample},
    multi/.style = {fill=orange!25},
    gedge/.style = {gray!55, line width=0.9pt},
]

\draw[gedge] (1.2,3.2) -- (0.0,1.9)  node[midway, sloped, above, font=\scriptsize, black] {$1$};
\draw[gedge] (1.2,3.2) -- (2.2,1.9)  node[midway, sloped, below, font=\scriptsize, black] {$2$};
\draw[gedge] (1.2,3.2) -- (4.0,1.6)  node[midway, sloped, above, font=\scriptsize, black] {$10$};
\draw[gedge] (0.0,1.9) -- (2.2,1.9)  node[midway, above, font=\scriptsize, black] {$6$};
\draw[gedge] (0.0,1.9) -- (1.1,0.6)  node[midway, sloped, above, font=\scriptsize, black] {$1$};
\draw[gedge] (2.2,1.9) -- (1.1,0.6)  node[midway, sloped, above, font=\scriptsize, black] {$3$};
\draw[gedge] (6.8,3.2) -- (5.2,2.0)  node[midway, sloped, above, font=\scriptsize, black] {$1$};
\draw[gedge] (6.8,3.2) -- (6.6,1.9)  node[midway, right, font=\scriptsize, black] {$1$};
\draw[gedge] (6.8,3.2) -- (8.2,1.9)  node[midway, sloped, above, font=\scriptsize, black] {$2$};
\draw[gedge] (5.2,2.0) -- (4.0,1.6)  node[midway, sloped, above, font=\scriptsize, black] {$1$};
\draw[gedge] (4.0,1.6) -- (4.0,0.1)  node[midway, right, font=\scriptsize, black] {$1$};
\draw[gedge] (6.6,1.9) -- (6.9,0.6)  node[midway, right, font=\scriptsize, black] {$1$};
\draw[gedge] (8.2,1.9) -- (6.9,0.6)  node[midway, sloped, above, font=\scriptsize, black] {$1$};
\draw[gedge] (1.1,0.6) -- (4.0,0.1)  node[midway, sloped, above, font=\scriptsize, black] {$4$};
\draw[gedge] (4.0,0.1) -- (6.9,0.6)  node[midway, sloped, above, font=\scriptsize, black] {$5$};

\node[svert] at (1.2,3.2) {$s_1$};
\node[svert] at (6.8,3.2) {$s_2$};
\node[vert] at (0.0,1.9)  {$1$};
\node[vert] at (2.2,1.9)  {$1$};
\node[vert] at (5.2,2.0)  {$1$};
\node[vert] at (6.6,1.9)  {$1$};
\node[vert] at (8.2,1.9)  {$1$};
\node[vert] at (1.1,0.6)  {$2$};
\node[vert] at (6.9,0.6)  {$2$};
\node[vert, multi] at (4.0,1.6) {$1,2$};
\node[vert, multi] at (4.0,0.1) {$2,3$};

\end{tikzpicture}
\caption{%
    Execution of \ActiveBellmanFord (\Cref{alg:active-bf}). Samples $S=\{s_1,s_2\}$ are settled in round $0$, and each number is the round in which a vertex's tentative distance was updated.
    Notice that orange vertices are updated twice, since a shorter path from $s_2$ corrects the initial estimate from $s_1$.
}
\label{fig:active-bf-rounds}
\end{figure}
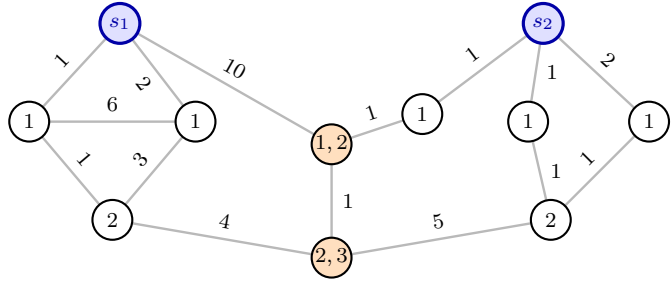

We remark that \Cref{alg:active-bf} is precisely active-vertex Bellman--Ford run from an auxiliary super-source $r$ in $G'=(V\cup\{r\},E\cup\{(r,s,0)\mid s\in S\})$: initializing $d[s]=0$ for every $s\in S$ corresponds to relaxing the edges $(r,s)$, and each iteration of the while loop corresponds to one round of relaxations.

Throughout, we count the initialization as round $0$ (corresponding to relaxing the implicit weight-$0$ edges $(r,s)$ for $s\in S$), and the $i$-th iteration of the while loop as round $i$.

An important property of the \ActiveBellmanFord algorithm is that it does not require knowledge of $ck$ in advance. The algorithm terminates when $\mathit{Active}=\emptyset$. We will show that at this point all values $\delta(v,b(v))$ have been correctly computed.
We remark that the restriction to active vertices is crucial for efficiency. 
Relaxing the edges of every vertex in every round would cost $O(m)$ per round, and the number of rounds is $\max_{v\in V}{|P(v,b(v))|}$, which w.h.p.\ is $\Omega(k\log n)$.
This gives $\Omega(mk\log n)$ overall, a $\log{n}$ factor more than our $O(mk)$ goal. 
By restricting relaxation only to active vertices, we get:

\begin{algorithm2e}[t]
\caption{$\ActiveBellmanFord(G,S)$}\label{alg:active-bf}
\DontPrintSemicolon
\lForEach{$v\in V$}{$d[v]\leftarrow\infty$}
\ForEach{$s\in S$}{$d[s]\leftarrow 0$\; $b[s]\leftarrow s$}
$\mathit{Active}\leftarrow S$\;
\While{$\mathit{Active}\neq\emptyset$}{
    $\mathit{NextActive}\leftarrow\emptyset$\;
    \ForEach{$v\in\mathit{Active}$}{
        \ForEach{$(v,w)\in E$}{
            \If{$d[v]+\wt(v,w)<d[w]$}{
                $d[w]\leftarrow d[v]+\wt(v,w)$\;
                $b[w]\leftarrow b[v]$\;
                $\mathit{NextActive}\leftarrow\mathit{NextActive}\cup\{w\}$\;
            }
        }
    }
    $\mathit{Active}\leftarrow\mathit{NextActive}$\;
}
\Return $\{d[v]\}_{v\in V}$, $\{b[v]\}_{v\in V}$\;
\end{algorithm2e}
\thmActiveBF*
\begin{proof}
    We first sketch the intuition. The analysis is a charging argument. After round $|P(v,b(v))|$, the tentative distance of $v$ is final, so $v$ is active in at most $|P(v,b(v))|$ rounds. The path $P(v,b(v))$ lies inside $\Ball(v)$, so $v$ is active at most $|\Ball(v)|$ times, and each activation relaxes its $\deg(v)$ incident edges. Summing $|\Ball(v)|\cdot\deg(v)$ over all vertices, and bounding each $\deg(v)$ by $O(m/n)$, gives the claimed running time $O\!\left(\sum_{v\in V}|\Ball(v)|\cdot m/n\right)$.

    
    First, we prove, by induction on $i$, that after round $i$ we have $d[v]\le\delta^{\le i+1}_{G'}(r,v)$ for every $v\in V$.
    For the base case, after round $0$ we have $d[s]=0=\delta^{\le 1}_{G'}(r,s)$ for every $s\in S$, and $d[v]=\infty=\delta^{\le 1}_{G'}(r,v)$ for every $v\notin S$.

    For the inductive step, we assume that after round $i-1$ we have $d[v]\le\delta^{\le i}_{G'}(r,v)$ for every $v\in V$ and show that after round $i$ we have $d[v]\le\delta^{\le i+1}_{G'}(r,v)$ for every $v\in V$.
    Let $(x,v)$ be the last edge of $P^{\le i+1}_{G'}(r,v)$. 

    By the induction hypothesis for $x$, after round $i-1$ we have $d[x]\le\delta^{\le i}_{G'}(r,x)$, and $d[x]$ attained its current value in some round $j\le i-1$;
    therefore, in round $j+1$ the vertex $x$ was active and the edge $(x,v)$ was relaxed, so after round $j+1$ we have $d[v]\le d[x]+\wt(x,v)\le\delta^{\le i}_{G'}(r,x)+\wt(x,v)=\delta^{\le i+1}_{G'}(r,v)$, as required.
    
    Since $d[v]$ is always the length of some path from $r$ to $v$ in $G'$, we have that $d[v]\ge\delta_{G'}(r,v)$.
    Therefore, after round $i$ we have $\delta_{G'}(r,v)\le d[v]\le\delta^{\le i+1}_{G'}(r,v)$, for every $v\in V$.
    By definition, if $|P_{G'}(r,v)|=i+1$ then $\delta^{\le i+1}_{G'}(r,v)=\delta_{G'}(r,v)$. Thus, $d[v]$ reaches its final value $\delta_{G'}(r,v)=\delta(v,b(v))$ after at most $|P_{G'}(r,v)|-1$ rounds.
    Once this happens, $b[v]=b(v)$ by \Cref{L-labels}, which is stated for exactly this in-place execution. 
    Grouping the vertices by $b[v]$ then yields $\Bundle(s)$ for every $s\in S$ in $O(n)$ additional time.

    Having established correctness, it remains to analyze the running time, governed by the number of rounds each vertex is active. 
    Since the path $P_{G'}(r,v)$ is the concatenation of the zero-weight edge $(r,b(v))$ with the path $P_G(b(v),v)$, and the graph is undirected, $|P_{G'}(r,v)|=1+|P_G(v,b(v))|$.

    Since every $x\in P_G(v,b(v))$ with $x\neq b(v)$ satisfies $\delta(v,x)\prec\delta(v,b(v))$ (item~\ref{P-sp-prefix} of \Cref{P-shortest-paths}), and since $b(v)\in\Ball(v)$ by definition, we have $P_G(v,b(v))\subseteq \Ball(v)$ and therefore $|P_G(v,b(v))|\le|\Ball(v)|$.

    Thus, after round $|P_G(v,b(v))|=O(|\Ball(v)|)$ we have $d[v]=\delta(v,b(v))$, and afterwards $v$ cannot be active anymore.
    Hence, $v$ is in $\mathit{Active}$ at most $O(|\Ball(v)|)$ times, and each time the algorithm triggers the relaxation of its $\deg(v)$ incident edges.

    Summing over all vertices, and using the fact that the degree is bounded by $O(m/n)$, the total running time is
    $$O\!\left(\sum_{v\in V}|\Ball(v)|\cdot \frac{m}{n}\right),$$
    as required.
\end{proof}
For a random sample, \Cref{L-sampling} gives $\sum_{v\in V}|\Ball(v)|=O(nk)$ with high probability, and substituting this into \Cref{T-active-BF} yields our desired running time of $O(mk)$.
\begin{restatable}{corollary}{corActiveBF}\label{C-active-BF-whp}
    If $S$ is formed by sampling each vertex of $V$ independently with probability $1/k$, where $\log^{0.25}{n}\le k\le\log^{0.99}{n}$, then $\ActiveBellmanFord(G,S)$ (\Cref{alg:active-bf}) runs in $O(mk)$ time with high probability.
\end{restatable}

We are now ready to present the complete algorithm and prove \Cref{T-Improved-undirected}. 
The algorithm $\ImprovedSSSP$ is composed of three steps. First, it sets $k=\log^{1/2}n\cdot(\log\log{n})^{-1/4}\cdot(\log\log\log{n})^{-1/4}$ and forms $S$ by sampling each vertex independently with probability $1/k$ and adding the source $u$.
\begin{enumerate}
    \item\label{step:bundles-} Run $\ActiveBellmanFord(G,S)$ (\Cref{alg:active-bf}) to compute $\delta(v,b(v))$ and $b(v)$ for every $v\in V$, and the bundles $\Bundle(s)$ for every $s\in S$ (\Cref{T-active-BF}).
    \item\label{step:balls-} Run $\BoundSSSP(v,B)$ for every $v\in V$ with bound $B=\delta(v,b(v))$ to obtain $\Ball(v)$.
    \item\label{step:bundle-dijkstra-} Call \BundleDijkstra of \cite{duan2023randomized} with the bundles computed in step~\ref{step:bundles-} and the balls computed in step~\ref{step:balls-}, and return the resulting shortest-path distances.
\end{enumerate}

\theoremImprovedUndirected*
\begin{proof}
    The correctness of \ImprovedSSSP follows directly from \Cref{T-active-BF} (bundles and $b(v)$), \Cref{L-Faster-D-SSSP} (balls), and \Cref{L-BundleDijkstra} (\BundleDijkstra).

    Notice that $\log^{0.25}n\le k\le \log^{0.99}n$, and therefore, w.h.p.\ all the properties of \Cref{L-sampling} hold. Adding $u$ to $S$ increases $|S|$ by one and can only shrink the balls, so these properties also hold for $S$. The running time of the algorithm is:
    \begin{enumerate}
        \item Computing $\delta(v,b(v))$ for every $v\in V$ and the bundles takes \[O\!\left(\sum_{v\in V}|\Ball(v)|\cdot\frac{m}{n}\right)\] time (\Cref{T-active-BF}), which is $O(mk)$ w.h.p.\ by \Cref{C-active-BF-whp}.
        \item Running $\BoundSSSP(v,B)$ for every $v\in V$ with bound $B=\delta(v,b(v))$ takes\\ $O(\frac{m}{n}\sum_{v\in V}|\Ball(v)|\log^{1/2}|\Ball(v)|(\log\log|\Ball(v)|)^{1/2})$ time (\Cref{L-Faster-D-SSSP}). By the third property of \Cref{L-sampling}, this is $O(mk\log^{1/2}{k}(\log\log{k})^{1/2})$ w.h.p.
        \item Computing SSSP by calling \BundleDijkstra (\Cref{L-BundleDijkstra}) takes \[O\!\left(\frac{m}{n}\sum_{v\in V}|\Ball(v)| + \frac{n}{k}\log{n}\right)\] time, which is $O(mk + (n/k)\log{n})$ w.h.p.\ by \Cref{L-sampling}.
    \end{enumerate}

    Combining the cost of the three steps, we get that the total running time w.h.p.\ is
    $$O\!\left(mk\log^{1/2}{k}(\log\log{k})^{1/2}+\frac{n}{k}\log{n}\right),$$
    and the choice of $k$ balances the two terms and yields the desired running time
    $$O\!\left(m\log^{1/2}n(\log\log{n})^{1/4}(\log\log\log{n})^{1/4}\right).$$
\end{proof}

\section{Bounding Bundle and Ball Sizes w.h.p.}\label{A-sampling}

\begin{lemma}[Bounded-dependence concentration~\cite{duan2023randomized}]\label{L-dependence}
    Let $\{Z_v\}_{v\in I}$ be random variables such that for each $v\in I$, $\EE[Z_v]=\mu$ and $Z_v\in[0,\beta]$ with probability $1$. Suppose each $Z_v$ is determined by a fixed deterministic set $W_v$, in the sense that whenever $W_{v_1},\dots,W_{v_j}$ are pairwise disjoint the variables $Z_{v_1},\dots,Z_{v_j}$ are independent, and each $W_v$ intersects at most $T$ of the other $W_u$. Then with probability at least $1-8T\beta\mu^{-1}e^{-\mu^3|I|/(8\beta^3T)}$,
    $$\sum_{v\in I}Z_v=\Theta(|I|\,\mu).$$
\end{lemma}

\lemmaSampling*

\begin{proof}
Fix an arbitrary constant $c\ge 1$; we show that all three properties hold simultaneously with probability at least $1-O(n^{-c})$.
Recall that $G$ is connected (\Cref{S-prelim}), and that by the degree reduction of \Cref{S-degree-reduction} (using $m=O(n)$) we may assume that the maximum degree of $G$ is at most a constant $\Delta\ge 3$.
Let $K=\ceil{(c+1)k\ln n}$ and $K'=\ceil{12k\ln k}$, and note that $12k\ln k\le K'\le 13k\ln k$ and $K'\le K\le n$ for all sufficiently large $n$.

\paragraph{Size of $S$.}
Let $X_v=\mathbf{1}[v\in S]$ be the indicator that $v$ is sampled; then $X_v\sim\mathrm{Bernoulli}(1/k)$ independently, and $|S|=\sum_{v\in V}X_v$, so $\EE[|S|]=n/k$. By the multiplicative Chernoff bound, $\Pr[|S|\ge 2n/k]\le e^{-n/(3k)}$; since $k\le\log^{0.99}{n}$, we have $n/k=\omega(\log n)$, so $|S|=O(n/k)$ with probability $1-n^{-\omega(1)}$.

\paragraph{The rank of the nearest sample.}
For $v\in V$, let $N_0(v)=\emptyset$, and for $1\le t\le n$, let $N_t(v)$ be the set of the $t$ closest vertices to $v$ (including $v$ itself); this is well defined since $G$ is connected and, under the total order of \Cref{S-positive-weights}, the distances from $v$ are pairwise distinct.
Let $R_v=\min\{t\ge 1 : N_t(v)\cap S\neq\emptyset\}$ be the rank of the sample nearest to $v$ (with $R_v=\infty$ if $S=\emptyset$).
Whenever $S\neq\emptyset$, since $b(v)$ is the sample nearest to $v$ and every $w\in\Ball(v)\setminus\{b(v)\}$ satisfies $\delta(v,w)\prec\delta(v,b(v))$, we have $\Ball(v)=N_{R_v}(v)$ and hence $|\Ball(v)|=R_v$.
Moreover, since $N_t(v)$ consists of exactly $t$ vertices, each sampled independently with probability $1/k$,
$$\Pr[R_v>t]=(1-1/k)^{t}\qquad\text{for every }0\le t\le n.$$

\paragraph{The radius event.}
For every $v\in V$, $\Pr[R_v>K]=(1-1/k)^{K}\le e^{-K/k}\le n^{-(c+1)}$.
By a union bound, with probability at least $1-n^{-c}$, every $v\in V$ satisfies $R_v\le K$; we call this the \emph{radius event}.

\paragraph{Bounded dependence.}
We apply \Cref{L-dependence} twice, both times with index set $I=V$ and the sets $W_v=N_{K'}(v)$, which are deterministic (the graph and its distances are fixed) and of size $|W_v|=K'$.
We first bound the number $T$ of other sets $W_u$ intersecting a fixed $W_v$, following the dependence argument of~\cite{duan2023randomized}.
Let $u\in V$ and $w\in N_{K'}(u)$. For every $x\in P_G(u,w)$ with $x\neq w$, we have $\delta(u,x)\prec\delta(u,w)$ by item~\ref{P-sp-prefix} of \Cref{P-shortest-paths}, and therefore $x\in N_{K'}(u)$ as well. Hence $P_G(u,w)\subseteq N_{K'}(u)$, so $P_G(u,w)$ has at most $K'$ vertices, i.e., at most $K'-1$ edges.
Consequently, for a fixed $w\in V$, every $u$ with $w\in W_u$ lies within hop-distance $K'-1$ of $w$, and since the maximum degree is at most $\Delta$, there are at most $\sum_{j=0}^{K'-1}\Delta^{j}\le\Delta^{K'}$ such vertices $u$.
Summing over the $K'$ vertices $w\in W_v$, we conclude that
$$T\le K'\cdot\Delta^{K'}=2^{O(k\log k)}=2^{O(\log^{0.99}n\cdot\log\log n)}=n^{o(1)},$$
where we used $k\le\log^{0.99}{n}$.

\paragraph{Sum of ball sizes.}
For $v\in V$ define $Z_v=\min(R_v,K')$ and $H_v=\mathbf{1}[R_v>K']$.
Since $R_v>t$ if and only if $N_t(v)\cap S=\emptyset$, both $Z_v$ and $H_v$ are functions of $S\cap W_v$ alone; hence pairwise disjoint sets $W_v$ yield independent variables $Z_v$ (respectively, $H_v$), as required by \Cref{L-dependence}.

For the variables $Z_v$ we have $Z_v\in[0,K']$ with probability $1$, and, since $Z_v>t$ if and only if $R_v>t$ for every $0\le t<K'$,
\[\begin{aligned}
\EE[Z_v]&=\sum_{t=0}^{K'-1}\Pr[Z_v>t]=\sum_{t=0}^{K'-1}\Pr[R_v>t]\\
&=\sum_{t=0}^{K'-1}(1-1/k)^{t}=k\left(1-(1-1/k)^{K'}\right)=:\mu,
\end{aligned}\]
which is identical for every $v\in V$; since $(1-1/k)^{K'}\le e^{-K'/k}\le k^{-12}\le 1/2$, we get $\mu=\Theta(k)$.
Substituting $\mu=\Theta(k)$, $\beta=K'=O(k\log k)$, $|I|=n$, and $T=n^{o(1)}$ into \Cref{L-dependence}, the exponent is
$$\frac{\mu^3|I|}{8\beta^3T}=\frac{\Theta(k^3)\cdot n}{O(k^3\log^3 k)\cdot n^{o(1)}}=n^{1-o(1)},$$
and the prefactor is $8T\beta\mu^{-1}=n^{o(1)}$, so $\sum_{v\in V}Z_v=\Theta(n\mu)=\Theta(nk)$ with probability $1-e^{-n^{1-o(1)}}$.

For the variables $H_v$ we have $H_v\in[0,1]$ and $\EE[H_v]=(1-1/k)^{K'}=:\mu_H$, identical for every $v\in V$. Using $1-x\ge e^{-2x}$ for $x\in[0,1/2]$ and $12k\ln k\le K'\le 13k\ln k$, we get $k^{-26}\le\mu_H\le k^{-12}$; in particular $\mu_H^{-1}=O(\log^{26}{n})$ and, since $k\ge\log^{0.25}{n}$, also $\mu_H\le\log^{-3}{n}$.
Substituting $\mu_H$, $\beta=1$, $|I|=n$, and $T=n^{o(1)}$ into \Cref{L-dependence}, the exponent is $\mu_H^3n/(8T)\ge n/(O(\log^{78}{n})\cdot n^{o(1)})=n^{1-o(1)}$ and the prefactor is $8T\mu_H^{-1}=n^{o(1)}$, so
$$\sum_{v\in V}H_v=\Theta(n\mu_H)=O\!\left(\frac{n}{\log^{3}n}\right)\qquad\text{with probability }1-e^{-n^{1-o(1)}}.$$

For the rest of the proof, fix the intersection of the radius event with the two concentration events above; by a union bound, it has probability at least $1-n^{-c}-2e^{-n^{1-o(1)}}=1-O(n^{-c})$.
On this intersection, for every $v\in V$, if $H_v=0$ then $|\Ball(v)|=R_v=Z_v$, and if $H_v=1$ then $|\Ball(v)|=R_v\le K$ by the radius event; in both cases $|\Ball(v)|\le Z_v+K\cdot H_v$. Therefore
$$\sum_{v\in V}|\Ball(v)|\;\le\;\sum_{v\in V}Z_v+K\sum_{v\in V}H_v\;=\;O(nk)+O\!\left(k\log n\cdot\frac{n}{\log^{3}n}\right)\;=\;O(nk),$$
proving the second property.

\paragraph{Weighted sum of ball sizes.}
Let $f(x)=\log^{1/2}(x)(\log\log(x))^{1/2}$. By the convention $\log{x}=\max\{\log_2{x},\,1\}$ (\Cref{S-prelim}), both $f$ and $x\mapsto x f(x)$ are non-decreasing on $[1,\infty)$, $f(x)\ge 1$, and $f(x)\le\log x$ (as $\log\log x\le\log x$).
We work on the same intersection of events as in the previous paragraph.
If $H_v=0$ then $|\Ball(v)|=Z_v\le K'$, so by monotonicity $|\Ball(v)|f(|\Ball(v)|)\le Z_v\,f(K')$; if $H_v=1$ then $|\Ball(v)|\le K$, so $|\Ball(v)|f(|\Ball(v)|)\le Kf(K)\cdot H_v$. Summing over $v\in V$,
$$\sum_{v\in V}|\Ball(v)|\,f(|\Ball(v)|)\;\le\;f(K')\sum_{v\in V}Z_v\;+\;Kf(K)\sum_{v\in V}H_v.$$
For the first term, $\log K'\le\log k+\log(13\ln k)=\log k+O(\log\log k)=O(\log k)$, and consequently $\log\log K'=O(\log\log k)$; hence $f(K')=O\!\left(\log^{1/2}(k)(\log\log(k))^{1/2}\right)=O(f(k))$, and the first term is $O\!\left(nk\,f(k)\right)$.
For the second term, $K=O(k\log n)=O(\log^{2}n)$, so $f(K)\le\log K=O(\log\log n)$, and therefore
\[\begin{aligned}
Kf(K)\sum_{v\in V}H_v&=O\!\left(k\log n\cdot\log\log n\cdot\frac{n}{\log^{3}n}\right)\\
&=O\!\left(\frac{nk\log\log n}{\log^{2}n}\right)=o(nk)=O\!\left(nk\,f(k)\right).
\end{aligned}\]
Combining the two terms gives
\[\begin{aligned}
&\sum_{v\in V}|\Ball(v)|\log^{1/2}|\Ball(v)|(\log\log|\Ball(v)|)^{1/2}\\
&=O\!\left(nk\log^{1/2}(k)(\log\log(k))^{1/2}\right),
\end{aligned}\]
as required.

Finally, a union bound over the failure events of all three properties gives total failure probability $O(n^{-c})$; since $c\ge 1$ was an arbitrary constant, all three properties hold with high probability.
\end{proof}

\section{From \texorpdfstring{$\BMSSP$}{BMSSP} to \texorpdfstring{$\BoundSSSP$}{BoundSSSP}}\label{A-BoundSSSP}

In this section we prove \Cref{L-Faster-D-SSSP}, deriving $\BoundSSSP$ from the recursive procedure $\BMSSP$ of~\cite{duan2026faster}.
\begin{reminder}{Reminder of \Cref{L-Faster-D-SSSP}}
    Let $G$ be a weighted directed graph with maximum degree $\Delta\ge 1$, let $v\in V$, let $B$ be a bound, and let $U=\{w\mid \delta(v,w)\prec B\}$.
    After a single global initialization taking $O(n)$ time, there is an algorithm $\BoundSSSP(v,B)$ that computes $\delta(v,w)$ for every $w\in U$ in $O(\Delta\,|U|\log^{1/2}{|U|}(\log\log|U|)^{1/2})$ time.
    In particular, on a graph preprocessed by the degree reduction of \Cref{S-degree-reduction}, in which $\Delta=m/n$, the running time is $O(\frac{m}{n}|U|\log^{1/2}{|U|}(\log\log|U|)^{1/2})$.
\end{reminder}
Throughout, $G$ is a directed graph with maximum degree $\Delta\ge1$. In our application, this degree bound is guaranteed by the reduction of \Cref{S-degree-reduction}; we view the resulting undirected graph as directed by replacing each edge with two opposite arcs, preserving distances and the degree bound up to a constant factor.
We fix the source $v$ and the bound $B$, and write $U=\{w\in V\mid\delta(v,w)\prec B\}$ for the target set.

\subsection{The Interface of \texorpdfstring{$\BMSSP$}{BMSSP}}

As in~\cite{duan2026faster}, an execution maintains a label $d[w]$ for every $w\in V$, where each finite $d[w]$ represents a concrete path from $v$ to $w$, so that $\delta(v,w)\preceq d[w]$; the vertex $w$ is \emph{complete} when $d[w]=\delta(v,w)$.
For a bound $B$ and a set $S\subseteq V$, the \emph{target} of the pair is
$$\widetilde U(B,S)=\{w\in V\mid \delta(v,w)\prec B\text{ and } P(v,w)\cap S\neq\emptyset\},$$
and a pair $\langle X,Y\rangle$ of vertex sets is a \emph{frontier} for a set $Z$ if every $z\in Z$ either belongs to $X$ and is complete, or has a complete vertex of $Y$ on $P(v,z)$.
The recursion of~\cite{duan2026faster} is controlled by an integer parameter $t\ge2$ and a level $\ell\ge0$; we write
$$\Lambda_\ell=t^3 2^{\ell t}$$
for the workload of a level-$\ell$ call.
With this notation, Lemma 3.1 of~\cite{duan2026faster} reads as follows.

\begin{lemma}[{$\BMSSP$; Lemma~3.1 of~\cite{duan2026faster}}]\label{L-BMSSP}
    Let $t\ge2$, let $0\le\ell\le\ceil{\log n/t}$, let $B$ be a bound, and let $S\subseteq V$ with $|S|\le t^2 2^{\ell t}$ be such that $\langle\emptyset,S\rangle$ is a frontier for $\widetilde U(B,S)$.
    Then $\BMSSP(B,S,\ell)$ returns a triple $(B',U',\mathcal D)$ with $B'\preceq B$ such that:
    \begin{enumerate}
        \item\label{I-bmssp-complete} $U'=\widetilde U(B',S)\subseteq\widetilde U(B,S)$ and every $w\in U'$ is complete, that is, $d[w]=\delta(v,w)$;
        \item\label{I-bmssp-size} $|U'|=O(\Lambda_\ell)$;
        \item\label{I-bmssp-full} if $B'=B$ (a \emph{full execution}) then $\mathcal D=\emptyset$ and $U'=\widetilde U(B,S)$;
        \item\label{I-bmssp-partial} if $B'\prec B$ (a \emph{partial execution}) then $|U'|=\Theta(\Lambda_\ell)$; we write $c_0\in(0,1]$ for the absolute constant such that a partial execution always returns $|U'|\ge c_0\Lambda_\ell$.
    \end{enumerate}
    The call runs in $O\bigl(|U'|(\ell\log t+\Delta t)\bigr)$ time.
\end{lemma}

We use two properties of \Cref{L-BMSSP}.
First, only a full execution certifies that the whole target has been computed. This is the reason~\cite{duan2026faster} invoke $\BMSSP$ at the top level with $S=\{s\}$, $B=\infty$ and $\ell=\ceil{\log n/t}$, a level at which $\Lambda_\ell\ge t^3n$ is sufficiently large and therefore a partial execution is impossible by item~\ref{I-bmssp-partial} of \Cref{L-BMSSP}.

Second, the running time is governed by the two competing terms $\ell\log t$ and $\Delta t$: taking $\ell$ as small as possible while still forcing a full execution, and then balancing the two terms, is what produces the bound of \Cref{L-Faster-D-SSSP}.
Using the top-level choice $\ell=\ceil{\log n/t}$ would make the running-time bound depend on $\log n$. In our application, $U=\Ball(v)\setminus\{b(v)\}$ can be much smaller than $V$, so we choose the parameters to obtain a bound in terms of $|U|$ instead.

\subsection{The Algorithm}
Since $|U|$ is unknown when $\BoundSSSP(v,B)$ is invoked, we search for the right parameters, doubling $t$ and increasing $\ell$ until a full execution occurs. See \Cref{alg:bound-sssp}.

\begin{algorithm2e}[t]
\caption{$\BoundSSSP(v,B)$}\label{alg:bound-sssp}
\DontPrintSemicolon
$t_0\leftarrow\max\{2,\ceil{(2/c_0)^{1/3}}\}$\tcp*{a constant; $c_0$ is the constant of item~\ref{I-bmssp-partial} of \Cref{L-BMSSP}}
\For{$j\leftarrow 0,1,2,\dots$}{
    $t\leftarrow 2^{j}t_0$\;
    $L\leftarrow\min\left\{\ceil{\frac{\Delta t}{\log t}},\ \ceil{\frac{\log n}{t}}\right\}$\;
    \For{$\ell\leftarrow 0$ \KwTo $L$}{
        restore $d[\cdot]$ to its initial state: $d[v]\leftarrow0$ and $d[w]\leftarrow\infty$ for $w\neq v$\;
        $(B',U',\mathcal D)\leftarrow\BMSSP(B,\{v\},\ell)$ \tcp*{run with parameter $t$}
        \lIf{$B'=B$}{\Return $\{(w,d[w])\}_{w\in U'}$}
    }
}
\end{algorithm2e}

The restoration in the body of the loop is not performed by scanning $V$, which would cost $\Omega(n)$ per call and $\Omega(n^2)$ over the $n$ invocations of $\BoundSSSP$ in \Cref{S-construction}.
Instead, the array $d[\cdot]$ is allocated and initialized to $\infty$ once, in $O(n)$ time. Before each call to $\BMSSP$, we set $d[v]=0$ for the current source $v$ and record all entries modified by the call, including $d[v]$. After copying any output distances, we reset these entries to $\infty$, at $O(1)$ time per modification, so the same array can be reused for the next call, even with a different source.
As each modification is charged to the call that made it, this at most doubles the running time of every call, and the $O(n)$ initialization is paid once, as stated in \Cref{L-Faster-D-SSSP}.

\subsection{Correctness}

\begin{claim}\label{C-precondition}
    Every call $\BMSSP(B,\{v\},\ell)$ made by \Cref{alg:bound-sssp} meets the preconditions of \Cref{L-BMSSP}, and $\widetilde U(B,\{v\})=U$.
\end{claim}
\begin{proof}
    Every path starting at $v$ contains $v$, so $\widetilde U(B,\{v\})=\{w\mid\delta(v,w)\prec B\}=U$.
    The call is made from the initial label state, in which $d[v]=0=\delta(v,v)$, so $v$ is complete; hence for every $z\in\widetilde U(B,\{v\})$ the path $P(v,z)$ contains the complete vertex $v$ of $\{v\}$, i.e., $\langle\emptyset,\{v\}\rangle$ is a frontier for $\widetilde U(B,\{v\})$.
    Finally, $|\{v\}|=1\le t^2 2^{\ell t}$, and $\ell\le L\le\ceil{\log n/t}$ by the definition of $L$.
\end{proof}

\begin{claim}\label{C-forced-full}
    A call $\BMSSP(B,\{v\},\ell)$ with $c_0\Lambda_\ell>|U|$ is a full execution.
    In particular, the call at level $\ell=\ceil{\log n/t}$ is a full execution, for every $t\ge t_0$.
\end{claim}
\begin{proof}
    Suppose the call is partial. By item~\ref{I-bmssp-partial} of \Cref{L-BMSSP} it returns $|U'|\ge c_0\Lambda_\ell>|U|$, while by item~\ref{I-bmssp-complete} of \Cref{L-BMSSP} and \Cref{C-precondition} we have $U'\subseteq\widetilde U(B,\{v\})=U$, a contradiction.
    For the second part, at $\ell=\ceil{\log n/t}$ we have $\Lambda_\ell\ge t^3 2^{\log n}=t^3n$, so $c_0\Lambda_\ell\ge c_0t_0^3n\ge 2n>|U|$ by the choice of $t_0$.
\end{proof}

By \Cref{C-forced-full}, any phase with $L=\ceil{\log n/t}$ returns by level $L$. The outer loop therefore terminates: for sufficiently large $t$, we have $\ceil{\Delta t/\log t}\ge\ceil{\log n/t}$, so $L=\ceil{\log n/t}$.
When the algorithm returns, it does so from a full execution, so by item~\ref{I-bmssp-full} of \Cref{L-BMSSP} and \Cref{C-precondition} we get $U'=\widetilde U(B,\{v\})=U$, and by item~\ref{I-bmssp-complete} of \Cref{L-BMSSP} every $w\in U'$ satisfies $d[w]=\delta(v,w)$.
Hence \Cref{alg:bound-sssp} outputs exactly $\delta(v,w)$ for every $w\in U$, as required.

\subsection{Running Time}

We first bound the cost of a single phase.

\begin{claim}\label{C-phase-cost}
    Consider a phase $j$ of \Cref{alg:bound-sssp}, with parameter $t=2^jt_0$, and let $\ell^\ast$ be the last level executed in it. Then the phase costs $O\bigl(\Lambda_{\ell^\ast}\,\Delta t\bigr)$ time; moreover, if the phase returns, it costs $O\bigl(|U|\,\Delta t\bigr)$ time.
\end{claim}
\begin{proof}
    By \Cref{L-BMSSP}, the call at level $\ell$ costs $O(|U'_\ell|(\ell\log t+\Delta t))$ time, where $U'_\ell$ is the set it returns, and by the definition of $L$,
    $$\ell\log t+\Delta t\;\le\;L\log t+\Delta t\;\le\;\left(\frac{\Delta t}{\log t}+1\right)\log t+\Delta t\;=\;2\Delta t+\log t\;=\;O(\Delta t),$$
    where we used $\log t\le t\le\Delta t$.
    The phase therefore costs $O\bigl(\Delta t\sum_{\ell\le\ell^\ast}|U'_\ell|\bigr)$.
    Since $t\ge2$, the workloads grow geometrically:
    $$\sum_{\ell=0}^{\ell^\ast}\Lambda_\ell=t^3\sum_{\ell=0}^{\ell^\ast}2^{\ell t}\le\frac{2^t}{2^t-1}\,\Lambda_{\ell^\ast}\le 2\Lambda_{\ell^\ast}.$$
    Combining this with $|U'_\ell|=O(\Lambda_\ell)$ (item~\ref{I-bmssp-size} of \Cref{L-BMSSP}) gives the first bound.
    For the second bound, suppose the phase returns at level $\ell^\ast$. If $\ell^\ast=0$, the only call is a full execution and costs $O(|U|\Delta t)$. Otherwise, every call at a level $\ell<\ell^\ast$ was then a partial execution, so items~\ref{I-bmssp-partial} and~\ref{I-bmssp-complete} of \Cref{L-BMSSP} give $c_0\Lambda_\ell\le|U'_\ell|\le|\widetilde U(B,\{v\})|=|U|$, and in particular $\Lambda_\ell\le|U|/c_0$ for every $\ell<\ell^\ast$. Hence
    $$\sum_{\ell\le\ell^\ast}|U'_\ell|\;=\;|U'_{\ell^\ast}|+\sum_{\ell<\ell^\ast}|U'_\ell|\;=\;|U|+O\Bigl(\sum_{\ell<\ell^\ast}\Lambda_\ell\Bigr)\;=\;|U|+O(\Lambda_{\ell^\ast-1})\;=\;O(|U|),$$
    where $|U'_{\ell^\ast}|=|U|$ because the last call is a full execution.
\end{proof}

\begin{claim}\label{C-last-phase}
    Let $J$ be the phase in which \Cref{alg:bound-sssp} returns. Then $t_J:=2^Jt_0=O\left(1+\sqrt{\frac{\log|U|\log\log|U|}{\Delta}}\right)$.
\end{claim}
\begin{proof}
    If $J=0$ then $t_J=t_0=O(1)$ and there is nothing to prove, so assume $J\ge1$ and let $t=2^{J-1}t_0$ be the parameter of the preceding phase, which did not return; let $L$ be its last level.
    By \Cref{C-forced-full}, if $L=\ceil{\log n/t}$ then that phase would have returned, so $L=\ceil{\Delta t/\log t}$ and the call at level $L$ was a partial execution.
    By items~\ref{I-bmssp-partial} and~\ref{I-bmssp-complete} of \Cref{L-BMSSP}, $c_0\Lambda_L\le|U'_L|\le|U|$, that is, $t^3 2^{Lt}\le|U|/c_0$, whence
    $$\frac{\Delta t^2}{\log t}\;\le\;Lt\;\le\;\log\frac{|U|}{c_0}\;=\;O(\log|U|),$$
    using the convention $\log x=\max\{\log_2 x,1\}$ of \Cref{S-prelim}.
    Since $t\ge\log t$ for every $t\ge1$, the left-hand side is at least $\Delta t\ge t$, so $t=O(\log|U|)$ and therefore $\log t=O(\log\log|U|)$.
    Substituting back, $\Delta t^2=O(\log|U|\log\log|U|)$, i.e., $t=O\bigl(\sqrt{\log|U|\log\log|U|/\Delta}\bigr)$, and $t_J=2t$.
\end{proof}

We now combine the phase bounds to prove \Cref{L-Faster-D-SSSP}.
Let $J$ be the phase in which the algorithm returns.
Every phase $j<J$ ran all of its levels without returning, so its last level was $L$ and, as in the proof of \Cref{C-last-phase}, the call at that level was partial and satisfied $\Lambda_L\le|U|/c_0$; by \Cref{C-phase-cost} that phase cost $O(|U|\Delta t_j)$ time.
Phase $J$ cost $O(|U|\Delta t_J)$ time, again by \Cref{C-phase-cost}.
Since the parameters $t_j=2^jt_0$ double, the total running time is
\[\begin{aligned}
&\sum_{j=0}^{J}O\bigl(|U|\Delta t_j\bigr)\;=\;O\bigl(|U|\Delta t_J\bigr)\\
&\;=\;O\Bigl(|U|\Delta+|U|\sqrt{\Delta\log|U|\log\log|U|}\Bigr),
\end{aligned}\]
by \Cref{C-last-phase}.
Finally, $\Delta\ge1$ implies $\sqrt{\Delta}\le\Delta$, and $\log^{1/2}|U|(\log\log|U|)^{1/2}\ge1$ by the convention of \Cref{S-prelim}, so both terms are bounded by \[O\bigl(\Delta|U|\log^{1/2}|U|(\log\log|U|)^{1/2}\bigr),\] which proves \Cref{L-Faster-D-SSSP}.
\qed

\fi
\end{document}